\documentclass[a4paper,12pt]{article}

\usepackage{mathrsfs}
\usepackage{graphicx}
\usepackage{amsmath}
\usepackage{amsfonts}
\usepackage{amssymb}
\usepackage{amsthm}
\usepackage{textcomp}
\usepackage{color}
\usepackage{geometry}
\usepackage[normalem]{ulem}
\usepackage{placeins}
\usepackage{bm}
\usepackage{wasysym}
\usepackage{titlesec}
\usepackage[makeroom]{cancel}

\newtheorem{theorem}{Theorem}[section]
\newtheorem{lemma}{Lemma}[section]
\newtheorem{remark}{Remark}[section]

\renewcommand{\theequation}{\arabic{section}.\arabic{equation}}
\renewcommand{\thetheorem}{\arabic{section}.\arabic{theorem}}
\renewcommand{\thelemma}{\arabic{section}.\arabic{lemma}}
\renewcommand{\theproposition}{\arabic{section}.\arabic{proposition}}

\renewcommand{\thealgorithm}{\arabic{section}.\arabic{algorithm}}

\newcommand\dd{~{\rm d}}
\newcommand{\burg}{{\sf b}}

\newcommand\R{\mathbb{R}}

\newcommand\Z{\mathbb{Z}}
\newcommand\C{\mathbb{C}}

\newcommand\Wc{\dot{\mathscr{W}}^{\rm c}}
\newcommand\WR{\dot{\mathscr{W}}_R}

\newcommand\D{\nabla}
\newcommand\del{\delta}

\newcommand{\<}{\langle}
\renewcommand{\>}{\rangle}

\newcommand{\ulin}{u^{\rm lin}}

\newcommand{\rzz}{\mathscr{R}_z}

\newcommand\Rc{R_{\rm c}}
\newcommand\Rb{R_{\rm b}}
\newcommand{\Ne}{N_{\rm e}}
\newcommand{\NeR}{N_{{\rm e}, R}}
\newcommand{\Rcore}{{R_{\rm def}}}
\newcommand{\Adm}{\mathscr{A}}
\newcommand{\Admu}{{\rm Adm}}
\newcommand{\AdmR}{{\rm Adm}(R)}

\def\D{\mathcal{D}}
\def\L{\Lambda}
\def\Lp{\Lambda^{\#}}
\def\Lhom{\Lambda^{\rm hom}}
\def\muhom{\mu_{\#}}

\def\XXint#1#2#3{{\setbox0=\hbox{$#1{#2#3}{\int}$ }
		\vcenter{\hbox{$#2#3$ }}\kern-.6\wd0}}

\def\asRC{{\bf (R)}}
\def\UsH{\dot{\mathscr{W}}^{1,2}}

\def\Ey{E}
\def\Eu{\mathcal{E}}

\def\Gy{G}
\def\Gu{\mathcal{G}}

\def\Ny{N}
\def\Nu{\mathcal{N}}

\def\Fy{F}
\def\Fu{\mathcal{F}}

\def\T{\mathcal{T}}
\def\J{\mathcal{J}}
\def\Ay{O}
\def\Au{\mathcal{O}}
\def\ay{\mathfrak{o}}
\def\Nn{N_{\Omega}}

\title{Thermodynamic Limit of Crystal Defects with Finite Temperature Tight Binding}

\author{Huajie Chen
	\footnote{{\tt huajie.chen@warwick.ac.uk}.
		School of Mathematical Sciences, Beijing Normal University, Beijing 100875 China.
		This work was supported by ERC Starting Grant 335120.}
	, Jianfeng Lu
	\footnote{{\tt jianfeng@math.duke.edu}.
      Departments of Mathematics, Physics and Chemistry, Duke University,
      Box 90320, Durham, NC 27708 USA. This work was supported in part by
      the National Science Foundation under grants DMS-1312659
      and DMS-1454939.}
	~and Christoph Ortner
	\footnote{{\tt c.ortner@warwick.ac.uk}.
		Mathematics Institute, University of Warwick, Coventry CV47AL UK.
		This work was supported by ERC Starting Grant 335120.}}
\date{}

\begin{document}
\maketitle

\begin{abstract}
	We consider a tight binding model for localised crystalline defects with
	electrons in the {\em canonical ensemble} (finite Fermi temperature)
	and nuclei positions relaxed according to the Born--Oppenheimer approximation.
	We prove that the limit model as the
	computational domain size grows to infinity is formulated in the {\em
	grand-canonical ensemble} for the electrons. The Fermi-level for the limit
	model is fixed at a homogeneous crystal level, independent of the defect or
	electron number in the sequence of finite-domain approximations. We
	quantify the rates of convergence for the nuclei configuration and for the
	Fermi-level.
\end{abstract}

\section{Introduction}
\label{sec:intro}
Electronic structure calculations based on density functional theory
and related models have been established as a predictive approach to
model a wide range of systems with important scientific and
engineering applications \cite{finnis03,martin04}.
Unlike empirical interatomic potentials, electronic structure models assume no
prior information on the chemical environment or atomic configuration,
which makes them a popular tool to model materials with defects \cite{vandewalle14}.

It is most natural to think of crystalline materials with defects
(vacancies, interstitials, dislocations, etc.) as an extended system with
an infinite number of atoms and electrons, however, in practical computations only finite
systems may be treated. It is therefore important to understand the
approximation error due to the choice of computational domain. In this paper we
start from the most common model for finite crystalline systems with defects,
characterise the limit model, and quantify the rate of convergence.


Thermodynamic limit problems have been studied at great length in the
analysis literature. The perfect lattice was studied in \cite{catto98}
for the Thomas--Fermi--von Weizs\"{a}cker (TFW) model and in
\cite{catto01} for the reduced Hartree--Fock (rHF) model. Results on
local defects in crystals in the
framework of the TFW and rHF models are \cite{blanc07, cances13b,
cances08a, cances08b, cances11, gontier15, lieb77}.
These discussions are restricted to the case where the nuclei are fixed on a
periodic lattice (or with a given local defect). Considering the
simultaneous relaxation of nuclei positions is a case of great physical
and mathematical interest. First steps in this direction have been taken in
 \cite{nazar14} for the Thomas--Fermi--von Weizs\"{a}cker model
and in \cite{chen15a} for a tight binding model under the simplifying
assumption of a ``fixed Fermi level".

A related problem is the continuum limit of quantum models. The TFW
models are studied in \cite{blanc02} where it is shown that, in the
continuum limit, the difference between the energies of the atomistic
and continuum models obtained using the Cauchy--Born rule tends to
zero. The tight binding and Kohn--Sham models are studied in a series
of papers \cite{e07, e10, e11, e13}, which establish the extension of
the Cauchy--Born rule to electronic structure for smoothly deformed
crystals. The macroscopic dielectric properties in the thermodynamic
limit are studied in the rHF model in \cite{cances10}.

For electronic structure models of defects, not only do we need to
consider the truncation of nuclei degrees of freedom and the
associated boundary conditions applied on atom positions, but more
importantly, we also have to restrict to finite number of electrons on
a finite domain. In particular, for systems with defects, it is a
priori unclear how many electrons should be imposed on the
computational domain, due to the relaxation of the electronic
structure. This can become a subtle issue especially when charged
defects are considered. (However, we do not treat charged defects in
the present work.)

In the present work, we investigate material defects in the context of
tight binding models, which are minimalist models for electronic
structure calculations. We assume that the electrons are in finite
temperature, whereas we adopt the Born-Oppenheimer approximation for
the nuclei (i.e., the nuclei degrees of freedom are under zero
temperature). Thus, the relaxation of electronic structure is
formulated in the {\rm canonical ensemble}, while atom positions are
determined by minimizing the free energy associated with the
electrons; see \eqref{problem_min_y}. For the electronic degree of
freedom, vacuum boundary condition with a buffer zone is assumed;
while a Dirichlet (clamped) boundary condition is employed for the atom
positions. Alternatively, periodic boundary conditions are considered
in \ref{sec:pbc}.

Our main results, formulated in Theorems~\ref{theorem:limit_problem_1}
and \ref{theorem:limit_problem_2}, state that the limiting problem is
formulated in the {\em grand-canonical ensemble} for the electrons, with
the Fermi-level of a homogeneous crystalline solid. Thus,
the limit problem is independent of the details of the defect and it is
in particular independent of how many electrons we impose on the finite domain.


Our results partially justify the ``fixed Fermi level'' approach for the
finite system that has been widely used, e.g., in \cite{cances08a,gontier15} for
the reduced Hartree--Fock model, in \cite{chen15a,pexsisigma} for the tight binding model,
and in \cite{chenlu15} for density functional theory. We may take the grand
canonical ensemble for the finite system with Fermi level given by the
perfect crystal. In the thermodynamic limit, the finite systems with fixed
Fermi level converge to the same infinite system. We use the qualifier
{\em partially} because the models mentioned above (except \cite{chen15a,pexsisigma})
are formulated at zero Fermi temperature, whereas our results treat
the case of finite Fermi temperature; see the Conclusion for further discussion.


A key ingredient in our analysis is a notion of locality of the electronic
structure model, which was also used in \cite{chen15a} to exhibit locality of
the potential energy, and which we extend here to other physical quantities, specifically
the number of electrons. Roughly
speaking, the dependence of a local physical property such as the local density
of states and hence local physical quantities on the environment decays
exponentially fast away from the physical location of interest. Therefore, away
from the boundary of the finite domain, the electronic structure behaves as that
of the infinite problem. A subtlety arises for the canonical ensemble as the
Fermi level for the finite system depends globally on the atom configuration,
which would destroy the locality. The key idea to overcome this difficulty is to view 
the Fermi-level as an independent variable, which together
with the nuclei positions, solves the constraint for the number of electrons
together with the force balance equation. The thermodynamic limit can then be
viewed as the convergence of the solution to the coupled system as the domain
tends to infinity.

\subsubsection*{Outline}

In Section \ref{sec:model_finite} we introduce the tight binding model for finite systems.
We discuss a `two-centre' tight binding model, set in the canonical ensemble
and the grand-canonical ensemble respectively,
and establish the strong locality of the local density of states.
In Section \ref{sec:point_tdl} we consider an infinite lattice with a local point defect.
We first derive the thermodynamic limits of the local density of states
by fixing the Fermi level, then present the convergence of the Fermi level,
and finally justify the thermodynamic limits of the finite problem with
certain boundary conditions.
In Section \ref{sec:conclusion}, we make concluding remarks and discuss future
perspectives. All the proofs are gathered in Section \ref{sec:proofs}.

In \ref{sec:pbc} we extend the analysis to point defects with periodic boundary
conditions and in \ref{sec:dislocation} to a straight dislocation line with
clamped boundary conditions.

\subsubsection*{Notation}
We will use the symbol $\langle\cdot,\cdot\rangle$ to denote an abstract duality
pairing between a Banach space and its dual.
We will use the Dirac bra-ket notation, which is widely used in quantum mechanics.
The notation defines the ``ket'' vector $|\psi\rangle$,
and its conjugate transpose called the ``bra'' vector $\langle\psi|$.

The symbol $|\cdot|$ normally denotes the Euclidean or Frobenius norm, 
while $\|\cdot\|$ denotes an operator norm.
For the sake of brevity of notation, we will denote $A\backslash\{a\}$ by
$A\backslash a$, and $\{b-a~\vert ~b\in A\}$ by $A-a$.
For $E \in C^2(X)$, the first and second variations are denoted by
$\<\delta E(u), v\>$ and $\<\delta^2 E(u) w, v\>$ for $u,v,w\in X$.

The symbol $C$ denotes a generic positive constant that may change from one line
of an estimate to the next. When estimating rates of decay or convergence, $C$
will always remain independent of the system size, of lattice position or of
test functions. The dependencies of $C$ will normally be clear from the context
or stated explicitly.

%

\section{The Tight Binding Model}
\label{sec:model_finite}
\setcounter{equation}{0}
\subsection{Free energy}
\label{sec:tb_FT}
Consider a many particle system consisting of $\Nn$ nuclei and $\Ne$ electrons.
Let $d\in\{1, 2,3\}$ be the space dimension and $\Omega\subset\R^d$ be an
{\it index set} or {\it reference configuration} with $\#\Omega=\Nn$.
An (atomic) {\em configuration} is a map $y : \Omega\to\R^d$ satisfying
\begin{equation} \label{eq:non-interpenetration}
	|y(\ell)-y(k)| \geq \mathfrak{m}|\ell-k| \qquad\forall~\ell,k\in\Omega
\end{equation}
with {\em accumulation parameter} $\mathfrak{m} > 0$.
In the following, we use $r_{\ell m}:=|y(\ell)-y(m)|$ for brevity of notation.

The tight binding model is a minimalist electronic structure model, which
enables the investigation and prediction of properties of molecules and
materials. For simplicity of presentation, we consider a `two-centre' tight
binding model \cite{goringe97,Papaconstantopoulos15}
(where the off-diagonal entries of the Hamiltonian are given by a pair potential)
with the identity overlap
matrix, and a single atomic orbital per atom. (The latter restriction requires
us also to assume that $0 < \Ne < 2 \Nn$.) Our results can be extended directly to
general non-self-consistent tight binding models with multiple atomic orbitals
per atom \cite[\S~2 and App.~A]{chen15a}.

The `two-centre' tight binding model is formulated in terms of a discrete
Hamiltonian, with the matrix elements
\begin{eqnarray}\label{tb-H-elements}
	\Big(\mathcal{H}(y)\Big)_{\ell k}
	= \left\{ \begin{array}{ll} 
	h_{\rm ons}\left(\sum_{j \neq \ell}
	\varrho(r_{\ell j})\right)
	& {\rm if}~\ell=k; \\[1ex]
	h_{\rm hop}(r_{\ell k}) & {\rm if}~\ell\neq k,
	\end{array} \right.
\end{eqnarray}
where
$h_{\rm ons} \in C^{\nu}((0, \infty))$ is the on-site term,
$\varrho \in C^{\nu}((0, \infty))$ with $\varrho = 0$ in
$[\Rc,\infty)$, where $\Rc$ a cut-off, $h_{\rm hop} \in
C^{\nu}((0, \infty))$ is the hopping term with
$h_{\rm hop}=0$ in $[\Rc,\infty)$ and  $\nu\geq 3$.
Note that $h_{\rm ons}$ and $h_{\rm hop}$ are independent of $\ell$ and $k$,
which indicates that all atoms of the system belong to the same species.

For future reference, we remark that the spectrum of $\mathcal{H}(y)$
is uniformly bounded in an interval $[\underline{\lambda}, \bar{\lambda}]$,
where $\underline{\lambda}, \overline{\lambda}$ depend only on $\mathfrak{m}$
but are independent of $\Omega$ or $y$ \cite[Lemma 2.1]{chen15a}.

The {\it Helmholtz free energy} (or {\it Mermin free energy}) of a system at absolute temperature $T>0$, as
a function of configuration $y$, is \cite{mermin65}
\begin{equation}\label{eq:Helmoltz_FE}
\begin{split}
	\Ey(y) := \min\Big\{ \mathfrak{F}\big(y,\{\psi_s\}_{s=1}^{\Nn},\{f_s\}_{s=1}^{\Nn}\big)
	~:~
	\psi_s:\Omega\to\R,~
	\psi_i^{\rm T}\psi_j=\delta_{ij}, & \\
	~ 0\leq f_s\leq 1, 
	\quad 2{\textstyle \sum_{s=1}^{\Nn}} f_s=\Ne & \Big\}
\end{split}
\end{equation}
where
\begin{align*}\label{eq:Mermin_FE}
	\mathfrak{F}\Big(y,\{\psi_s\}_{s=1}^{\Nn},\{f_s\}_{s=1}^{\Nn}\Big)
	&:= \sum_{s=1}^{\Nn}  \Big(
	 			2 f_s \big\< \psi_s \big| \mathcal{H}(y) \big| \psi_s \big\>
				+ 2k_{\rm B}T  S(f_s) \Big), \\
	S(f) &:= f \ln f + (1-f)\ln(1-f),
\end{align*}
$k_{\rm B}$ is Boltzmann's constant, and the factor $2$ comes from spin degeneracy.
For simplicity, we will write $\beta:=(k_{\rm B}T)^{-1}$ for the inverse temperature.
In the Helmholtz free energy \eqref{eq:Helmoltz_FE}, $f_s$ is understood as the occupation
number of the electronic state with orbital function $\psi_s$. Thus the occupation number
is between $0$ and $1$ according to Pauli's exclusion principle and the total number of
electrons is given by $N_e$ (counting spin degeneracy).

A straightforward calculation implies that there exists a minimizer
$\{\psi_s\}, \{f_s\}$ satisfying
\begin{equation}\label{eq:eigenpair_fs}
	\mathcal{H}(y)\psi_s=\lambda_s\psi_s
	\qquad{\rm and}\qquad
	f_s=\frac{1}{1+e^{\beta(\lambda_s-\mu)}}
	\qquad {\rm for}~ s=1,\cdots,\Nn,
\end{equation}
where the Lagrange multiplier $\mu$, known as the {\em chemical potential},
is chosen such that
\begin{equation}\label{eq:def_mu}
	2\sum_{s=1}^{\Nn} f_s = \Ne.
\end{equation}
Since the ordered eigenvalues $\{\lambda_s\}_{s = 1}^{\Nn}$ are fixed with given $y$ and the functional
\begin{displaymath}
	\Ny(y, \tau) := 2\sum_{s=1}^{\Nn}\left(1+e^{\beta(\lambda_s-\tau)}\right)^{-1}
\end{displaymath}
is strictly monotone and continuous in $\tau$ with $\Ny(y, \tau) \to 0$ (resp. $2\Nn$), as
$\tau \to - \infty$ (resp. $+\infty$), it follows that $\mu$ is uniquely defined.
Note that $\{\lambda_s\},  \{\psi_s\}, \{f_s\}$ and $\mu$ given above all depend on $y$,
however we suppress this dependence in the notation.

With $f(x) := 1/(1+e^{\beta x})$ we can now rewrite the Helmholtz free energy
\eqref{eq:Helmoltz_FE} as
\begin{align}\label{eq:Helmoltz_FE_eigen}
	\Ey(y) &= \sum_{s=1}^{\Nn} \mathfrak{e}(\lambda_s,\mu)
	\qquad{\rm with}\quad \\ \nonumber
	\mathfrak{e}(x,\tau) &:= 2x f(x-\tau) + \frac{2}{\beta} S\big(f(x-\tau)\big)
	\\ \label{eq:e_N+G}
	&= 2\tau f(x-\tau) + \frac{2}{\beta}\ln\big(1-f(x-\tau)\big),
\end{align}
where $\{\lambda_s\}_{s=1}^{\Nn}$ are eigenvalues of $\mathcal{H}(y)$ and $\mu$
satisfies \eqref{eq:def_mu} (the last equality is easily verfied;  see
also \cite{alavi94}). In particular, we have from the regularity assumptions on $h_{\rm ons}$ and $h_{\rm hop}$
that $\Ey$ is $\nu$ times continuously
differentiable (in the sense of Fr\'{e}chet) on the set of configurations $y$
satisfying \eqref{eq:non-interpenetration}.
To see this, we refer to \eqref{proof:force_HF} for the first order derivative calculation
(see similar calculations for higher order derivatives in \cite[(42)]{chen15a}).

\begin{remark}\label{rem:contft}
	While the solution to the variational problem
    \eqref{eq:Helmoltz_FE} in the case of the tight binding model is
    straightforward, the variational problem associated to the
    Helmholtz free energy in the continuous case is in fact quite
    subtle. The difficulty arises when the spectrum of the Hamiltonian
    operator contains a continuous spectrum part, such as the
    Hamiltonians for atoms, molecules and solids. (This is of course
    impossible for the discrete tight binding model.) Since there
    exists an infinite number of states below an energy level $\lambda_0$,
    by occupying $M$ such states with occupation number $\frac{1}{M}$
    for each state, the Helmholtz free energy is then smaller than $
    \lambda_0 - \ln M$, which goes to negative infinity as $M \to
    \infty$. Therefore, even for systems as simple as an atom, some
    renormalization is needed to make the variational formulation
    well-posed in the continuous case.
\end{remark}

\subsection{Equilibration of nuclei}
We consider stable equilibria (local minima) of $\Ey$ under a boundary
condition: for some $\Omega^{\rm D} \subset \Omega$ we seek
\begin{equation} \label{problem_min_y}
	\bar{y} \in \arg\min \big\{ \Ey(y) ~:~ y(\ell)=\ell
	\text{ for }\ell \in \Omega^{\rm D} \big\}.
\end{equation}
Specifically, we are interested in determining the limiting model of \eqref{problem_min_y} as $\Nn \to \infty$.  To that end, in the remainder of \S~\ref{sec:model_finite} we assemble some useful observations about the finite-$\Nn$ model, for which the precise choices of $\Omega$ and $\Omega^{\rm D}$ are unimportant, see Remark \ref{remark:Omega_D}.

Abusing notation, we write
\begin{align*}
   \Ey(y, \tau) := \sum_{s = 1}^{\Nn} \mathfrak{e}(\lambda_s, \tau),
\end{align*}
then the constrained minimisation problem
\begin{equation}\label{problem_min_y_tau}
	(\bar{y}, \bar{\mu}) \in \arg\min \Big\{ \Ey(y, \tau) ~:~
	\Ny(y, \tau) = \Ne, ~y(\ell) = \ell \text{ for }\ell \in \Omega^{\rm D}
	\Big\}
\end{equation}
is fully equivalent to \eqref{problem_min_y}. We will see in
\S~\ref{sec:local_prop} that this formulation is analytically convenient due
to the fact that $\Ey(y, \tau)$ and $\Ny(y, \tau)$ are separable as functions
of $y$, while $\Ey(y) = \Ey(y, \mu(y))$ contains a small amount of non-local
interaction due to the global dependence of $\mu$ on $y$.

	\begin{remark}\label{remark:Omega_D}
		The choice of $\Omega^{\rm D}$ is not unique, for example, one can even choose $\Omega^{\rm D}$ to be empty set.
		In later sections, we will take $\Omega^{\rm D}$ as an outer ``buffer layer'' sourrounding the atoms to be relaxed.
		For a given finite system  $\Omega$, the minima $\bar{y}$ of \eqref{problem_min_y} depend on the choice of $\Omega^{\rm D}$.
		However, the limiting problem (as $\Nn\rightarrow\infty$) will be independent of the
		choice of (sequence of) $\Omega^{\rm D}$.
		%
		To pass to the limit we will specify a concrete relation between $\Omega$ and $\Omega^{\rm D}$ in \S  \ref{sec:vf_limit}.
	\end{remark}

\subsection{The grand potential and other quantities of interest}
\label{sec:grand-potential}
The problem \eqref{problem_min_y}  is set on the
{\em canonical ensemble}, where the Helmholtz energy is minimized at
equilibrium with constant temperature and particle number.
By contrast we can also define an analogous problem in the
{\em grand-canonical ensemble}, where the chemical potential $\mu$ is
a fixed model parameter while the particle number $\Ne$ is variable.

For this situation, $\Ey$ is replaced with the {\it grand potential},
\begin{align}\label{def_grand_potential}
	&\Gy(y,\mu) := \Ey(y,\mu)-\mu N(y, \mu) =\sum_{s=1}^{\Nn}\mathfrak{g}(\lambda_s,\mu)
	\qquad{\rm with}
	\\ \nonumber
	& \mathfrak{g}(x,\tau):= \mathfrak{e}(x,\tau)-2\tau f(x-\tau) =
	\frac{2}{\beta}\ln\big(1-f(x-\tau)\big).
\end{align}

The energies $\Ey, \Gy$ and the particle number $\Ny$ are the three
main quantities of interest for our work.
Upon defining $\mathfrak{n}(x, \tau) := 2 f(x-\tau)$ the three quantities
$\Ey, \Gy$, and $\Ny$ are of the form
\begin{eqnarray}\label{eq:eigenpair-observable}
   \Ay(y, \tau) = \sum_{s = 1}^{\Nn} \ay(\lambda_s, \tau).
\end{eqnarray}
We call $\Ay$ an analytic {\em quantity of interest (QoI)} if there exists a strip 
$U = \{ a+ib : |b| < \mathfrak{d}~{\rm for~some}~\mathfrak{d}>0 \}$ 
such that $\ay(\cdot, \tau)$ is analytic on $U$ for all $\tau
\in \R$ and $\partial_z^j \ay(z, \tau)$ is continuous on $U \times \R$.
This is satisfied for $\Ay = \Ey, \Ny, \Gy$ with $\mathfrak{d} = \pi/\beta$.

\begin{remark} \label{rem:grand-can-frcs}
	Using the fact $\partial_x \mathfrak{g}(x,\tau)=2f(x-\tau)$ (see \cite{alavi94}), we have
	\begin{eqnarray}\label{force:G}
		\frac{\partial\Gy(y,\tau)}{\partial y(\ell)}
		= 2\sum_{s=1}^{\Nn} f(\lambda_s-\tau) \bigg\<
		\psi_s \bigg| \frac{\partial\mathcal{H}(y)}{\partial y(\ell)}
		\bigg| \psi_s \bigg\>
	\end{eqnarray}
	from a similar calculation as that in \eqref{proof:force_HF}.
	Thus, if $\tau = \mu$ is the chemical potential satisfying
	\eqref{eq:def_mu}, then a straightforward calculation
	(c.f. \eqref{proof:force_HF} or \cite[\S 7.6.2]{finnis03}) implies that
	\begin{eqnarray}\label{nablaG_nablaE}
		\frac{\partial\Gy(y,\tau)}{\partial y(\ell)} \Big|_{\tau = \mu}
		= \frac{\partial \Ey(y)}{\partial y(\ell)},
	\end{eqnarray}
	This connection between $\Ey$ and $\Gy$ is a key observation in
    our derivation of the thermodynamic limit of \eqref{problem_min_y}.
\end{remark}

\subsection{Spatial decomposition of analytic quantities of interest}
\label{sec:local_prop}
Assuming we have the eigenpairs $\{\lambda_s,\psi_s\}_{s=1}^{\Nn}$ of the
Hamiltonian $\mathcal{H}(y)$, it is useful to define the (total) {\it density of
states} \cite{finnis03} of the system by
\begin{eqnarray}\label{eq:DOS}
	\D(y,\epsilon) = \sum_{s=1}^{\Nn} \delta(\epsilon-\lambda_s).
\end{eqnarray}
This should be understood in the operational sense,  i.e.,
\begin{eqnarray*}
	\big\< \D(y), g \big\>
		= \int g(\epsilon) \D(y,\epsilon)\dd\epsilon
		= \sum_{s=1}^{\Nn}g(\lambda_s) \qquad \text{for } g \in C(\R).
\end{eqnarray*}
If $\Ay = \Ay(y, \tau)$ is an analytic QoI (in particular,
$\Ay=\Ey, \Gy, \Ny$), we can write
\begin{eqnarray*}
	\Ay(y,\tau) = \< \D(y), \ay(\cdot, \tau)\>
\end{eqnarray*}

A spatial decomposition of $\D$
would automatically lead to a spatial decomposition of $\Ay$,
which will be a powerful analytical tool. Following \cite{chen15a,
  ercolessi05, finnis03}, we can introduce the {\it local density of
  states} (or, {\em projected density of states}),
\begin{eqnarray}\label{def_LDOS}
	\D_{\ell}(y,\epsilon) := \sum_{s=1}^{\Nn} \delta(\epsilon-\lambda_s) \big[\psi_s\big]_{\ell}^2,
\end{eqnarray}
where
$[\psi_s]_{\ell}$ is the $\ell$-th entry of $\psi_s$.
Thus, we obtain a local variant of the analytic QoI $\Ay$,
\begin{eqnarray}\label{local_A_N}
	\Ay_{\ell}(y,\tau) :=  \big\<\D_{\ell}(y),\ay(\cdot,\tau)\big\> =
	\sum_{s=1}^{\Nn} \ay(\lambda_s,\tau) \big[\psi_s\big]_{\ell}^2.
\end{eqnarray}
It is easy to verify that
\begin{align*}
	\D(y) &= \sum_{\ell=1}^{\Nn}\D_{\ell}(y) \qquad \text{and hence} \qquad
	\Ay(y,\tau) = \sum_{\ell=1}^{\Nn}\Ay_{\ell}(y,\tau) .
\end{align*}
The next lemma states the locality of $\D_{\ell}(y,\cdot)$, which is the
backbone of our analysis.

\begin{lemma}\label{lemma:locality_fixed_mu}
	Let $y$ be an atomistic configuration with accumulation
	parameter $\mathfrak{m}$, and let $\Ay$ be an analytic QoI.
	Then, for $1\leq j\leq\nu$,
	there exist positive constants $C_j$ and $\gamma_j$
	depending only on $d,~\mathfrak{m},~\Rc,~h_{\rm hop},~h_{\rm ons},~\ay$
	and $\tau$, such that
	\begin{eqnarray}\label{eq:locality_ldos}
	\left|\frac{\partial^j \Ay_{\ell}(y,\tau)}
	{\partial [y(m_1)]_{i_1}\cdots\partial [y(m_j)]_{i_j}}\right|
	\leq C_j e^{-\gamma_j\sum_{t=1}^j r_{\ell m_t}}
	\end{eqnarray}
	for any $1\leq \ell\leq N_{\Omega}$, $1 \leq m_1, \cdots, m_j \leq N_{\Omega}$
	and $1\leq i_1,\cdots,i_j\leq d$.

	The constants $C_j$ are bounded above and $\gamma_j$ are bounded away from zero
	on bounded intervals for $\tau$.
\end{lemma}
\begin{proof}
	The proof is analogous to that of \cite[Lemma 2.3]{chen15a}, but for the sake
	of completeness, we present it in \S~\ref{sec:proof_locality}.
\end{proof}

\begin{remark}
	We emphasize that it is crucial to keep $\tau$ fixed to obtain this locality
	result. A $y$-dependent chemical potential would introduce a small amount of
	non-locality in $\Ay_{\ell}$, which is not easy to control directly. Note how
	in \eqref{problem_min_y_tau} we have split off this non-locality
	at the expense of adding a constraint to the system, cf. \eqref{problem_min_y}.
	However, since that constraint is given as a sum of local quantities it
	is convenient to treat analytically.
\end{remark}

\begin{remark}\label{remark:symmetry}
  With the definition of Hamiltonian \eqref{tb-H-elements}, we have
  isometry and permutation invariance of $\D_\ell$ and hence of the
  quantities $\Ay_\ell$ (for an analogous proof see \cite[Lemma
  2.4]{chen15a}): if $\mathcal{I}:\R^d\rightarrow\R^d$ is an isometry,
  then $\Ay_{\ell}(y,\tau)=\Ay_{\ell}(\mathcal{I}(y),\tau)$; if $\Pi$
  is a permutation of $\Omega$,
  then $\Ay_{\ell}(y,\tau)=\Ay_{\Pi^{-1}(\ell)}(y\circ\Pi,\tau)$.
\end{remark}

\begin{remark}
	As shown in the proofs of Lemma \ref{lemma:locality_fixed_mu}, 
	we can alternatively use a matrix-trace representation for the analytic QoIs.
	Instead of using the eigenpair formulations \eqref{eq:eigenpair-observable} 
	and \eqref{local_A_N}, 
	we can also write (c.f. \eqref{operator-trace} and \eqref{eq:FE_l_contour})
	\begin{eqnarray}
	\Ay(y,\tau) = {\rm Tr}\Big[ \ay\big(\mathcal{H}(y),\tau\big) \Big]
	\quad{\rm and}\quad
	\Ay_{\ell}(y,\tau) = \Big[ \ay\big(\mathcal{H}(y),\tau\big) \Big]_{\ell\ell} .
	\end{eqnarray}
	Indeed, this formulation will bring a lot of convenience to our analysis.
\end{remark}

\section{Thermodynamic limit of a crystal defect}
\label{sec:point_tdl}
\setcounter{equation}{0}
While, in \S~\ref{sec:model_finite}, we considered general atomistic
configurations $y$, we now focus on crystalline defects. For the sake of clarity
of presentation the main text concentrates on point defects. An extension to
straight dislocation lines is briefly discussed in \ref{sec:dislocation}.

Employing the separability of the various physical quantities established in
Lemma~\ref{lemma:locality_fixed_mu} we will formulate a model for a crystalline
defect in an infinite lattice and then prove that solutions of
\eqref{problem_min_y_tau} converge to a solution of the infinite lattice model.
Again, for the sake of clarity of presentation, the main text concentrates on
clamped boundary condition (Dirichlet boundary condition) for the finite size
systems, while  periodic boundary conditions are discussed in \ref{sec:pbc}.

\subsection{Reference configuration}
\label{sec:reference_space}
We consider a single defect embedded in an infinite homogeneous crystalline bulk.
A homogeneous crystal reference configuration is given by the Bravais lattice
$\Lhom=A\Z^d$, for some non-singular matrix $A \in \mathbb{R}^{d \times d}$. A
{\em point defect reference configuration} is a set $\L \subset \R^d$ satisfying
\begin{flushleft}\label{as:asRC}
	\asRC \quad
	$\exists ~\Rcore>0$, such that
	$\L\backslash B_{\Rcore} = \Lhom\backslash B_{\Rcore}$ and
	$\L \cap B_{\Rcore}$ is finite.
\end{flushleft}
Then the set of possible (atomic) {\em configurations} is
\begin{align*}
	\Adm_{0}(\L) &:= \bigcup_{\mathfrak{m}>0} \Adm_{\mathfrak{m}}(\L),
		\qquad \text{where} \\
	\Adm_{\mathfrak{m}}(\L) &:= \left\{ y:\L \rightarrow \R^{d}, ~
	|y(\ell)-y(m)| > \mathfrak{m} |\ell-m|
	\quad\forall~  \ell, m \in \L \right\},
\end{align*}
where we have again imposed accumulation parameter $\mathfrak{m}$.

\subsection{Limit of the local density of states}
\label{sec:limit_ldos}
We first study pointwise thermodynamic limits with a fixed chemical potential.
For $y \in \Adm_0$ and a finite subset $\Omega\subset\L$, we
define $y^{\Omega}:\Omega \to\R^d$, $y^\Omega(\ell) := y(\ell)$, $\forall\, \ell \in \Omega$.
The local density of states and local analytic QoIs associated with $y^\Omega$ will be denoted
by $\D_{\ell}(y^{\Omega},\cdot)$, $\Ay_{\ell}(y^\Omega, \tau)$,
where the dependence on $\Omega$ is implicitly assumed.

The following lemma establishes the existence of the (thermodynamic) limit of
$\D_{\ell}(y^{\Omega})$ as $\Omega\uparrow\L$, via local analytic QoI $\Ay_\ell$.
This result is
closely related to the locality result in Lemma \ref{lemma:locality_fixed_mu}.
We will skip the details of the proofs and refer to \cite[Theorem 3.1]{chen15a}
for an analogous argument.

\begin{lemma}[pointwise thermodynamic limit]
\label{lemma:limit_fixed_mu}
If $\L$ satisfies \asRC~and $y\in\Adm_0(\L)$,
then for any $\ell\in\L$ and for any sequence of bounded sets
$\Omega_R\supseteq B_R(\ell)$, the limit
\begin{eqnarray*}
	\Ay_{\ell}(y,\tau): = \lim_{R\to\infty} \Ay_{\ell}(y^{\Omega_R},\tau)
\end{eqnarray*}
exists with a fixed $\tau\in\R$
and is independent of the choice of sets $\Omega_R$.
Moreover, there exist constants $C_j$ and $\eta_j$ for $0\leq j\leq\nu$
such that
\begin{align}
\label{err_pointwise_limit}
& \big| \Ay_{\ell}(y,\tau) - \Ay_{\ell}(y^{\Omega_R},\tau) \big| \leq C_0e^{-\eta_0 R}
\qquad{\rm and} \\[1ex] \nonumber
& \left| \frac{\partial^j \Ay_{\ell}(y,\tau)}{\partial [y(m_1)]_{i_1}\cdots\partial [y(m_j)]_{i_j}}
- \frac{\partial^j \Ay_{\ell}(y^{\Omega_R},\tau)}{\partial [y(m_1)]_{i_1}\cdots\partial [y(m_j)]_{i_j}} \right|
\leq C_j e^{-\eta_j \big(R+\sum_{k=1}^j r_{\ell m_k}\big)}
\\ \label{err_d_pointwise_limit}
& \qquad\qquad \forall~m_k\in\L\cap\Omega_R ~~ 1\leq k\leq j, \quad \forall~1\leq i_1,\cdots,i_j\leq d .
\end{align}
\end{lemma}

Similar to Lemma \ref{lemma:locality_fixed_mu}, the constants in
Lemma \ref{lemma:limit_fixed_mu} depend
only on $d$, $\mathfrak{m}$, $\Rc$, $h_{\rm hop}$, $h_{\rm ons}$, $\ay$
and $\tau$,
but $C_j$ are bounded above and $\eta_j$ are bounded away from 0 on bounded intervals for $\tau$.

From \eqref{err_d_pointwise_limit} and the locality \eqref{eq:locality_ldos},
we can derive the locality of the thermodynamic limits:
if $\tau\in\R$ is fixed, then
\begin{eqnarray}\label{eq:locality_FE_limit}
\left|\frac{\partial^j \Ay_{\ell}(y,\tau)}{\partial [y(m_1)]_{i_1}\cdots\partial [y(m_j)]_{i_j}}\right|
\leq C_j e^{-\gamma_j\sum_{l=1}^j r_{\ell m_l}}
\end{eqnarray}
for any $\ell\in\L$ and $1\leq i_1,\cdots,i_j\leq d$.

\begin{remark}\label{remark:symmetry_limit}
From the isometry and permutation invariance of $\Ay_{\ell}(y^{\Omega_R},\tau)$,
we can also derive the isometry and permutation invariance of their thermodynamic
limits:
if $\mathcal{I}:\R^d\rightarrow\R^d$ is an isometry, then
$\Ay_{\ell}(y,\tau)=\Ay_{\ell}(\mathcal{I}(y),\tau)$;
if $\Pi$ is a permutation of $\L$, then
$\Ay_{\ell}(y,\tau)=\Ay_{\Pi^{-1}(\ell)}(y\circ\Pi,\tau)$.
\end{remark}

\subsection{Energy space for displacements}
We can decompose the configuration $y$ into
\begin{eqnarray}\label{config_y_u}
	y(\ell) = \ell + u(\ell)  \qquad\forall~\ell\in\Lambda,
\end{eqnarray}
where $u:\Lambda\rightarrow\mathbb{R}^d$ is called the {\em displacement}.
Since we are considering point defects, it is natural to assume that
displacements in the infinite lattice model will belong to an
energy space \cite{chenpre_vardef,ehrlacher13}, which we define next.

If $\ell, \ell+\rho\in\Lambda$ then we define the finite difference
$D_\rho u(\ell) := u(\ell+\rho) - u(\ell)$. The full interaction stencil
is defined by $Du(\ell) := (D_{\rho}u(\ell))_{\rho \in \L-\ell}$.
For a stencil $Du(\ell)$ and $\gamma > 0$ we define the (semi-)norms
\begin{eqnarray*}
	\big|Du(\ell)\big|_\gamma := \bigg( \sum_{\rho \in \L-\ell} e^{-2\gamma|\rho|}
	\big|D_\rho u(\ell)\big|^2 \bigg)^{1/2}
	\quad{\rm and}\quad
	\| Du \|_{\ell^2_\gamma} := \bigg( \sum_{\ell \in \L}
	|Du(\ell)|_\gamma^2 \bigg)^{1/2}.
\end{eqnarray*}
We will also use the norm $\| Du \|_{\ell^1_\gamma} := \sum_{\ell \in \L} |Du(\ell)|_\gamma$ in our analysis.
For any $\Omega\subset\L$, we define
$\| Du \|_{\ell^p_\gamma(\Omega)} := \bigg( \sum_{\ell \in \Omega}
|Du(\ell)|_\gamma^p \bigg)^{1/p}$
with $p=1,2$.

We have from \cite{chenpre_vardef} that all (semi-)norms
$\|\cdot\|_{\ell^2_\gamma}$ with $\gamma>0$ are equivalent.
Following \cite{chenpre_vardef,ehrlacher13} we can therefore define the function
space of finite energy displacements,
\begin{displaymath}
	\UsH(\L) := \big\{ u:\L\to\R^d,
	~ \|Du\|_{\ell^2_\gamma}<\infty \big\},
\end{displaymath}
with the associated semi-norm $\|Du\|_{\ell^2_\gamma}$.

Upon defining $x:\L\rightarrow\R, x(\ell) :=\ell$, the associated class
of {\em admissible displacements} is
\begin{eqnarray*}
	\Admu(\L):= \big\{ u\in\UsH(\L) ~:~ x+u\in\Adm_0(\L) \big\}.
\end{eqnarray*}

We now transform local analytic QoIs $\Ay_\ell(y, \tau)$ to become functions of
displacements $u = y - x$. Due to the isometry (translation) invariance
(Remark \ref{remark:symmetry_limit}),
we may represent $\Ay_\ell$ as a function of $(Du(\ell),\tau)$, i.e.,
\begin{eqnarray*}
	\Au_{\ell}(Du(\ell),\tau):=\Ay_{\ell}(x+u,\tau) .
\end{eqnarray*}
Moreover, if $\L=\Lhom$, then permutation invariance
(Remark \ref{remark:symmetry_limit}) also removes
the dependence on the lattice site, i.e., we can write
\begin{equation} \label{eq:local-homogeneous-observable}
	\Au_{\ell}(Du(\ell),\tau) = \Au_{\#}(Du(\ell),\tau)
	\qquad \forall~\ell\in\Lhom.
\end{equation}

\subsection{Limit of the chemical potential}
\label{sec:limit_mu}
Before we state the variational problem on the limit lattice $\L$ we
investigate the behaviour of the chemical potentials as $\Omega \uparrow \L$.

With the notation \eqref{eq:local-homogeneous-observable} we can define the
{\it Fermi level} of a homogeneous crystal, $\muhom$, such that
\begin{eqnarray}\label{def:fermi_level}
	\Nu_{\#}(\pmb{0},\muhom) = 1 .
\end{eqnarray}
Note that $\muhom$ is uniquely defined since the thermodynamic limit
$\Nu_{\#}(\pmb{0},\tau)$ is a strictly monotone continuous function of $\tau$,
with $\Nu_{\#}(\pmb{0}, \tau) \to 0$ (resp. $2$) as $\tau \to -\infty$
(resp. $+\infty$). See also Remark~\ref{rem:bloch}.

\begin{theorem} \label{theorem:limit_mu}
	Let $\L$ satisfy \asRC, $\L_R := \L \cap B_R \uparrow \L$ and $N_R := \#
	\L_R$. For each $R$ let $u_R : \L_R \to \R^d$ with $y_R(\ell) := \ell+
	u_R(\ell)$ a {\em configuration} with parameter $\mathfrak{m}$ independent of
	$R$.

	Let $\NeR \in \mathbb{R}$ be a prescribed number of electrons in the subsystem
	$\L_R$, chosen such that $|N_R - \NeR|$ is bounded as $R \to \infty$.
	Then, for $R$ sufficiently large, the chemical potential $\mu_R$ solving
	$N(y_R,\mu_R) = \NeR$ is well-defined and satisfies
	\begin{eqnarray}\label{eq:mu_converge}
    \big|\mu_R - \muhom \big|
	\leq C_1 \left( R^{-1} + R^{-d} \| Du_R \|_{\ell^1_\gamma} \right)
	\leq C_2 \left( R^{-1} + R^{-d/2} \| Du_R \|_{\ell^2_\gamma} \right)
	\end{eqnarray}
	with some constants $C_1$ and $C_2$.
\end{theorem}

\begin{proof}
The proof is presented in \S~\ref{sec:proof_limit_mu}.
\end{proof}

Informally, the chemical potential converges to the Fermi level of the
corresponding homogeneous lattice if the displacement norm $\| Du_R
\|_{\ell^1_\gamma}$ grows slower than $N_R$. The interpretation of this result
is that the Fermi level only changes through a global (non-rigid) transformation
of the lattice structure, such as a change in the lattice constant.

\begin{remark} \label{rem:bloch}
One can alternatively define the Fermi level $\muhom$  using Bloch's theorem
\cite{kittel96}. Let $y_{\#}:\Lhom\to\R^d$ be such that $y_{\#}(\ell)=\ell$ and
${\rm BZ}$ be the first Brillouin zone associated with $\Lhom$. For any point
$k\in{\rm BZ}$, we have the associated Hamiltonian
(in the simple `two-centre' tight binding setting \eqref{tb-H-elements})
\begin{eqnarray*}
\mathcal{H}^{(k)}(y_{\#}) := h_{\rm ons}\left(\sum_{\ell\in\Lhom\backslash \{0\}} \varrho\big( r_{0\ell} \big)\right)
+ \sum_{\ell\in\Lhom\backslash \{0\}} \exp\big(ik\cdot \pmb{r}_{0\ell} \big) \cdot h_{\rm hop}\big(r_{0\ell} \big)
\end{eqnarray*}
with $\pmb{r}_{0\ell}=y_{\#}(\ell)-y_{\#}(0)$ and $r_{0\ell}=|\pmb{r}_{0\ell}|$.
Note that $\mathcal{H}^{(k)}(y_{\#})$ is an $1\times 1$ matrix since there is only one atom in each unit cell and one atomic orbital for each atom.
Then we can obtain the corresponding eigenvalue $\lambda^{(k)}=\mathcal{H}^{(k)}(y_{\#})$,
and the Fermi level $\muhom$ is defined such that
\begin{eqnarray}\label{eq:mu_per}
\int_{{\rm BZ}} f(\lambda^{(k)} - \muhom) \dd k=\frac{1}{2}.
\end{eqnarray}
Using \eqref{eq:mu_per},  $\muhom$ can be efficiently computed numerically.
The generalisation to multiple atomic orbitals is straightforward.
\end{remark}

\subsection{Main result: the thermodynamic limit}
\label{sec:vf_limit}
In view of Remark \ref{rem:grand-can-frcs} and Theorem~\ref{theorem:limit_mu}
the grand potential with
{\em fixed} Fermi-level $\muhom$ is a natural candidate for the limit energy.
Thus, for a displacement $u : \L \to \R^d$ we define (formally at first)
the {\it grand potential difference functional} of the infinite system by
\begin{eqnarray}\label{def_grand_diff}
	\Gu(u):=\sum_{\ell\in\L} \Big( \Gu_{\ell}(Du(\ell),\muhom)
	- \Gu_{\ell}(\pmb{0},\muhom) \Big).
\end{eqnarray}
Using locality \eqref{eq:locality_FE_limit} of $\Gu_\ell$
 we can obtain
the following result, which states that the difference functional is
well-defined. We refer to \cite{chenpre_vardef} (see also \cite{ehrlacher13})
for a rigorous proof.

\begin{lemma}\label{lemma:free_E_space}
	$\Gu$ is well-defined on $\Wc(\L) \cap \Admu(\L)$ where
	\begin{eqnarray*}
		\Wc(\L) = \left\{ u \in\UsH(\L),~\exists~R>0
		~s.t.~u={\rm const} ~in~ \L\backslash B_R \right\},
	\end{eqnarray*}
	and continuous with respect to the $\UsH$-topology. In particular, there
	exists a unique continuous extension to $\Admu(\L)$. The extended
	functional, still denoted by $\Gu$, is $\nu$ times Fr\'{e}chet
	differentiable.
\end{lemma}

The force equilibration problem associated with $\Gu$ is
\begin{eqnarray}\label{problem_min_inf}
	\bar{u}\in\arg\min\Big\{ \Gu(u)~:~u\in\Admu(\L) \Big\},
\end{eqnarray}
where ``$\arg\min$'' is understood in the sense of local minimality.
We claim that solutions to
\eqref{problem_min_y_tau} converge to a solution of
\eqref{problem_min_inf}. We will prove two complementary results to
establish this.

First, we reformulate
\eqref{problem_min_y_tau} in terms of displacements, and
specify a sequence of domains $\Omega$ and $\Omega^{\rm D}$.
For each domain radius $R > 0$ we choose a buffer radius $\Rb = \Rb(R)$
with $\Rb \to \infty$ as $R \to \infty$ and define
\begin{displaymath}
   \L_R := \L \cap B_{R+\Rb}  \qquad \text{and} \qquad
   \L_R^{\rm D} := \L_R \setminus B_{R}.
\end{displaymath}
The associated set of admissible displacements is
\begin{multline*}
	\AdmR := \Big\{	u\in\WR(\L) ~\big|~ x+u\in\Adm_0(\L)  \Big\},
\\[1ex]
{\rm with} \qquad
\WR(\L):=\Big\{ u:\L\rightarrow\R^d ~\big|~
u=0~{\rm in}~\L\backslash B_R \Big\} .
\qquad
\end{multline*}
Note that we have extended displacements of $\L_R$ by zero in order to be able
to estimate errors. The resulting finite-domain equilibrium problem
corresponding to \eqref{problem_min_y_tau}
(with $\Omega=\L_R$ and $\Omega^{\rm D} =\L_R^{\rm D}$)
is to find $(\bar{u}_R,\bar{\mu}_R)\in\AdmR\times\R$ such that
\begin{eqnarray}\label{equil_force_R}
(\bar{u}_R,\bar{\mu}_R)
\in \arg\min \Big\{ \Eu^{\L_R}(u_R, \tau) ~:~
\Nu^{\L_R}(u_R, \tau) = \NeR,~\tau\in\R, ~u_R\in\AdmR
\Big\} ,
\quad
\end{eqnarray}
where $\Eu^{\L_R}(u_R,\tau) := \Ey\big((x+u_R)|_{\L_R},\tau\big)$,
$\Nu^{\L_R}(u_R,\tau) := \Ny\big((x+u_R)|_{\L_R},\tau\big)$
and $\NeR$ is the number of electrons contained in $\L_R$.

As indicated above, we present two rigorous justifications of
\eqref{problem_min_inf}, the proofs of which are, respectively, given in
\S~\ref{sec:proof_theo_1} and \S~\ref{sec:proof_theo_2}. 
We refer to Remark \ref{remark:ass:main} for a discussion of the assumptions under which these results hold.

First, we show that, if \eqref{problem_min_inf} has a
solution $\bar{u}$, then there exist solutions $\bar{u}_R$ to
\eqref{equil_force_R} such that
$\| D\bar{u}_R - D\bar{u} \|_{\ell^2_\gamma} \to 0$.

\begin{theorem}\label{theorem:limit_problem_1}
	Assume that $|N_R - \NeR|$ is bounded as $R \to \infty$.
	If $\bar{u}\in\Admu(\L)$ is a solution of \eqref{problem_min_inf}
	which is also strongly stable, i.e.,
	\begin{eqnarray}\label{ubar_stability}
	\big\<\delta^2\Gu(\bar{u}) v, v\big\>
	\geq \bar{c}\|Dv\|^2_{\ell^2_{\gamma}} \qquad\forall~v\in\UsH(\L),
	\end{eqnarray}
	then there are constants $R_0, c_{\rm b} > 0$
	such that, for $R>R_0$ and $\Rb > c_{\rm b} \log R$,
	there exists a solution $(\bar{u}_R,\bar{\mu}_R)$ of \eqref{equil_force_R}
	satisfying
	\begin{eqnarray}\label{err_Du}
		\big\|D\bar{u}-D\bar{u}_R\big\|_{\ell^2_{\gamma}}
		+ \big|\bar{\mu}_R-\muhom\big| \leq CR^{-\min\{1,d/2\}}.
	\end{eqnarray}
\end{theorem}

Our second result reverses the argument: if $\bar{u}_R$ is a bounded
sequence of solutions to \eqref{equil_force_R}, then any accumulation
point $\bar{u}$ solves \eqref{problem_min_inf}.

\begin{theorem} \label{theorem:limit_problem_2}
	Let $R_j \uparrow \infty$ and $(\bar{u}_{R_j}, \bar{\mu}_{R_j})$ be
	solutions to \eqref{equil_force_R} with $R = R_j$,
	if $|N_{R_j} - N_{{\rm e},R_j}|$ is bounded
	and $\sup_{j > 0} \| D\bar{u}_{R_j} \|_{\ell^2_\gamma} < \infty$, then there
	exists a subsequence (not relabelled) and $\bar{u} \in \Admu(\L)$ such that
	\begin{eqnarray}\label{convergence_uj}
	\bar{\mu}_{R_j} \to \muhom \quad \text{and} \quad
	D_\rho \bar{u}_{R_j}(\ell) \to D_\rho \bar{u}(\ell) \quad
	\forall~\ell \in \L, ~\rho \in \L - \ell.
	\end{eqnarray}
	Moreover, each such accumulation point $\bar{u}$ solves \eqref{problem_min_inf}.
\end{theorem}

\begin{remark}\label{remark:ass:main}
	Theorem \ref{theorem:limit_problem_1} assumes the existence of a stable
	solution to the limit problem. This assumption is the natural generalisation
	of phonon stability \cite{kittel96} to defects, and from a physical
	perspective very mild. However we are not aware of any means to prove it
	rigorously; even in the context of classical interatomic potentials few
	results under very stringent assumptions exist \cite{garroni,hudson}, and
	indeed only for the case of anti-plane screw dislocations where a
	topologically imposed infinite energy barrier makes such an analysis tractable.

	On the other hand, Theorem \ref{theorem:limit_problem_2} assumes uniform
	boundedness of approximation solutions, which is a weaker uniform stability
	assumption placed on the the sequence of approximations. Again, we are
	unaware of any avenue to establish it rigorously, but it is interesting from
	a practical perspective since this assumption could be checked {\it a
	posteriori} during a numerical simulation.
\end{remark}

\section{Conclusions}
\label{sec:conclusion}
In this paper, we derive the thermodynamic limit for a coupled
electron and geometry relaxation problem in the context of the tight binding model for
crystalline defects. In particular, we have seen that the Fermi level of the finite systems converges to
the Fermi level of the homogeneous crystal in which the defect is embedded, and
that the equilibrium states of the finite system  converge to the minimizer of
the infinite grand potential.

A key motivation for our analysis is that it lends strong theoretical
support to the ``fixed Fermi-level'' assumption approach in recent
analyses of multi-scale methods.  The canonical ensemble setting, where the Fermi level
depends globally on the atom configuration, we cannot exploit locality
of electronic structure \cite{chen15a,chenlu15}. However, Theorems
\ref{theorem:limit_problem_1} and \ref{theorem:limit_problem_2}
indicate that we can approximate the canonical ensemble equilibrium
state by minimizing the grand potential with the (fixed) Fermi level
of the perfect crystal.  The strong locality results arising in this
setting then allow the construction and rigorous analysis of
linear-scaling, QM/MM multi-scale, and Green's function embedding
methods \cite{chen15b,chen15a,chenlu15,pexsisigma}.

The ``fixed Fermi-level'' assumption has also been employed in zero temperature
electronic structure models for insulators \cite{cances08a,chenlu15,gontier15}.
In the setting of crystalline defects it is not immediately clear
how to choose it. A possible choice would be through the zero temperature limit $\beta\rightarrow\infty$. This leads to the interesting issue
that the thermodynamic limit most likely does not commute with the
zero temperature limit, due to eigenstates in the band-gap which give rise to
$O(1)$ changes in the Fermi-level. Thus, the correct choice of Fermi-level
at zero (or low) Fermi temperature is an interesting and subtle issue.


A final key question is whether our result can be extended to the more
accurate electronic structure models, such as Kohn--Sham density
functional theory. The main difficulty is to control the long-range
Coulomb interaction, which gives rise to substantial technical and
conceptual challenges, in particular the possibility of charged
defects~\cite{cances08a}. Moreover the variational
formulation of the canonical ensemble also becomes subtle in the
continuous setting, as mentioned in Remark~\ref{rem:contft}.

\section{Proofs}
\label{sec:proofs}
\setcounter{equation}{0}

\subsection{Proof of Lemma \ref{lemma:locality_fixed_mu}}
\label{sec:proof_locality}

The analysis of the locality results in Lemma \ref{lemma:locality_fixed_mu}
builds on a representation of $\Ay(y,\tau)$ in terms of contour integrals.
This technique has been used in quantum chemistry,
for example \cite{chen15a,e10,goedecker95} for tight binding models.

For an atomic configuration $y$, we can rewrite $\Ay(y,\tau)$ as the trace of
some operator-valued function of the Hamiltonian
\begin{eqnarray}\label{operator-trace}
\Ay(y,\tau) = {\rm Tr}\Big[ \ay\big(\mathcal{H}(y),\tau\big) \Big].
\end{eqnarray}
Following \cite{chen15a}, we can find a bounded
contour $\mathscr{C} \subset \C$, circling all the eigenvalues $\lambda_s$
on the real axis and avoiding the intersection with the non-analytic region
of $\ay(\cdot,\tau)$ at the same time.
Then we have
\begin{eqnarray}\label{eq:O_contour}
\Ay(y,\tau) = -\frac{1}{2\pi i}\oint_{\mathscr{C}} \ay(z,\tau) {\rm Tr}
\Big[\big(\mathcal{H}(y)-zI\big)^{-1}\Big] \dd z .
\end{eqnarray}

We can also derive similar representations for $\Ay_{\ell}(y,\tau)
=\<\D_{\ell}(y,\cdot),\ay(\cdot,\tau)\>$.
Let $e_{\ell}$ be the $N$ dimensional canonical basis vector,
then we obtain from the definition \eqref{def_LDOS} that
\begin{align}\label{eq:FE_l_contour} \nonumber
\Ay_{\ell}(y,\tau)
&=\sum_{s=1}^{\Nn} \ay(\lambda_s,\tau)(\psi_s,e_{\ell})(e_{\ell}, \psi_s)
=\sum_{s=1}^{\Nn}\Big(\ay\big(\mathcal{H}(y),\tau\big)
\psi_s,e_{\ell}\Big) (e_{\ell}, \psi_s)
\\ \nonumber
&= \sum_{s=1}^{\Nn}(e_{\ell}, \psi_s) \Big(\psi_s,\ay\big(\mathcal{H}(y),\tau\big)
e_{\ell}\Big) = \Big(e_{\ell},\ay\big(\mathcal{H}(y),\tau\big)e_{\ell}\Big)
\\
&= -\frac{1}{2\pi i}\oint_{\mathscr{C}}
\ay(z,\tau)\Big[\big(\mathcal{H}(y)-zI\big)^{-1}\Big]_{\ell\ell}\dd z.
\end{align}

\begin{proof}
	[Proof of Lemma \ref{lemma:locality_fixed_mu}]
	\quad
	First, we have from the definition \eqref{tb-H-elements} that the
	Hamiltonian matrix $\mathcal{H}(y)$
	is {\it banded} in the sense that
	\begin{eqnarray*}
		\Big(\mathcal{H}(y)\Big)_{\ell k}=0 \qquad{\rm if}~r_{\ell k}\geq\Rc.
	\end{eqnarray*}
	Denoting the resolvent by $\rzz=\big(\mathcal{H}(y)-zI\big)^{-1}$,
	we have from \cite[Lemma 2.2]{chen15a} and \cite[Lemma 12]{e10} that
	there exist constants $C_{\rm a}$ and $\eta_{\rm a}$ such that
	\begin{eqnarray}\label{locality_resolvent}
	\left|\Big(\rzz(y)\Big)_{\ell k}\right| \leq C_{\rm a} e^{-\eta_{\rm a} r_{\ell k}}
	\qquad\forall~z\in\mathscr{C} ,
	\end{eqnarray}
	where $C_{\rm a}$ depends on $h_{\rm hop}$ and $h_{\rm ons}$, and
	$\eta_{\rm a}$ depends on $\mathfrak{m}$, $\Rc$ and $\ay$ (through $\mathfrak{d}$, which equals $\pi/\beta$ for $\ay=\mathfrak{e},\mathfrak{n},\mathfrak{g}$).
	For sake of readability, we will drop the argument $(y)$ in $\mathcal{H}(y)$
	and $\rzz(y)$ whenever convenient and possible without confusion.

	Denoting the first and second order partial derivatives of Hamiltonians by
	\begin{eqnarray*}
		\Big(\left[\mathcal{H}_{,m}(y)\right]_i\Big)_{\ell k}
		= \frac{\partial \big(\mathcal{H}(y)\big)_{\ell k}}{\partial [y(m)]_i}
		\quad{\rm and}\quad
		\Big(\left[\mathcal{H}_{,mn}(y)\right]_{i_1 i_2}\Big)_{\ell k}
		= \frac{\partial^2 \big(\mathcal{H}(y)\big)_{\ell k}}{\partial [y(m)]_{i_1} \partial [y(n)]_{i_2}}
	\end{eqnarray*}
	with $1\leq i, i_1,i_2\leq d$,
	we can calculate the first and second order derivatives of
	$\Ay_{\ell}(y,\tau)$ based on \eqref{eq:FE_l_contour},
	\begin{align}
	\label{eq:dFE_l}
	\frac{\partial \Ay_{\ell}(y,\tau)}{\partial [y(m)]_i}
	&= \frac{1}{2\pi i} \oint_{\mathscr{C}} \ay(z,\tau)
	\Big[\rzz\left[\mathcal{H}_{,m}\right]_i\rzz\Big]_{\ell\ell}\dd z
	\qquad{\rm and}
	\\[1ex] \nonumber
	\frac{\partial^2 \Ay_{\ell}(y,\tau)}{\partial [y(m)]_{i_1}\partial [y(n)]_{i_2}}
	&= \frac{1}{2\pi i}\oint_{\mathscr{C}} \ay(z,\tau)
	\bigg[ \mathscr{R}_z \big[\mathcal{H}_{,mn}\big]_{i_1 i_2} \mathscr{R}_z
	- \mathscr{R}_z \left[\mathcal{H}_{,m}\right]_{i_1}
	\mathscr{R}_z\big[\mathcal{H}_{,n} \big]_{i_2}  \mathscr{R}_z
	\\ \label{eq:d2FE_l}
	& \hspace{10em}
	- \mathscr{R}_z \left[\mathcal{H}_{,n} \right]_{i_2} \mathscr{R}_z
	\big[\mathcal{H}_{,m}\big]_{i_1} \mathscr{R}_z \bigg]_{\ell\ell} \dd z.
	\end{align}

	For $j=1$,
	we have from the definition \eqref{tb-H-elements} that
	\begin{eqnarray*}
		\left|\Big(\left[\mathcal{H}_{,m}\right]_i(y)\Big)_{\ell k}\right|
		\leq C_{\rm b} e^{-\eta_{\rm b}(r_{\ell m}+r_{km})},
	\end{eqnarray*}
	where the constant $C_{\rm b}$ depends on $h_{\rm hop}$ and $h_{\rm ons}$,
	and $\eta_{\rm b}$ depends on $\Rc$.
	This together with \eqref{locality_resolvent} implies
	\begin{multline}\label{proof:locality_ab}
	\Big[\rzz\left[\mathcal{H}_{,m}\right]_i\rzz \Big]_{\ell\ell}
	= \sum_{1\leq\ell_1,\ell_2\leq \Nn} \big[\rzz \big]_{\ell\ell_1}
	\big( \left[\mathcal{H}_{,m}\right]_i \big)_{\ell_1\ell_2}
	\big[\rzz\big]_{\ell_2\ell}
	\\[1ex]
	\leq C_{\rm a}^2 C_{\rm b} \sum_{1\leq \ell_1,\ell_2\leq \Nn} e^{-\min\{\eta_{\rm a},\eta_{\rm b}\}
	\big(r_{\ell\ell_1} + r_{\ell_1 m} + r_{m \ell_2} + r_{\ell_2 \ell}\big)}
	\leq C_{\rm a}^2 C_{\rm b} e^{-\min\{\eta_{\rm a},\eta_{\rm b}\} r_{\ell m}}
	. \qquad
	\end{multline}
	We then obtain from \eqref{eq:dFE_l} and \eqref{proof:locality_ab} that
	\begin{eqnarray*}
		\frac{\partial \Ay_{\ell}(y)}{\partial [y(m)]_i} \leq C_{\rm a}^2 C_{\rm b}
		|\mathscr{C}| \left( \sup_{z\in\mathscr{C}}|\ay(z,\tau)| \right)
		e^{-\min\{\eta_{\rm a},\eta_{\rm b}\} r_{\ell m}} \leq C_1 e^{-\eta_1 r_{\ell m}}
		\qquad{\rm for}~1\leq i\leq d ,
	\end{eqnarray*}
	where $|\mathscr{C}|$ depends on $d$, $\mathfrak{m}$, $h_{\rm hop}$,
	$h_{\rm ons}$ and $\Rc$, and $\sup_{z\in\mathscr{C}}|\ay(z,\tau)|$ depends
	on $\ay$ and $\tau$.
	This completes the proof for $j=1$.

	For $j=2$,
	we have from the definition \eqref{tb-H-elements} that
	\begin{eqnarray*}
		\left|\Big(\left[\mathcal{H}_{,mn}\right]_{i_1 i_2}(y)\Big)_{\ell k}\right|
		\leq Ce^{-\gamma (r_{\ell m} + r_{km} + r_{\ell n} + r_{kn})},
	\end{eqnarray*}
	which together with \eqref{locality_resolvent} implies
	\begin{align*}
	& \Big[ \rzz \left[ \mathcal{H}_{,m} \right]_{i_1} \rzz \left[ \mathcal{H}_{,n} \right]_{i_2} \rzz \Big]_{\ell\ell}
	\leq C e^{-\frac12 \gamma ( r_{\ell m}+r_{\ell n} ) };
	\\[0.5em]
	& \Big[\rzz \left[ \mathcal{H}_{,n} \right]_{i_2} \rzz \left[ \mathcal{H}_{,m} \right]_{i_1} \rzz \Big]_{\ell\ell}
	\leq C e^{-\frac12 \gamma (r_{\ell m}+r_{\ell n} )};
	\qquad \text{and}
	\\[0.5em]
	& \Big[ \rzz \left[ \mathcal{H}_{,mn} \right]_{i_1 i_2} \rzz \Big]_{\ell\ell}
	\leq C e^{-\frac12 \gamma (r_{\ell m}+r_{\ell n} )}.
	\end{align*}
	Inserting these three estimates into \eqref{eq:d2FE_l} yields the desired result,
	\begin{displaymath}
	\frac{\partial^2 \Ay_{\ell}(y)}{\partial [y(m)]_{i_1}\partial [y(n)]_{i_2}} \leq C_2
	e^{-\eta_2\big(r_{\ell m}+r_{\ell n}\big)} \qquad{\rm for}~1\leq i_1,i_2\leq d.
	\end{displaymath}

	We will skip the details for the proofs for cases $j\geq 2$,  which are analogous but tedious.
\end{proof}

\subsection{Proof of Theorem \ref{theorem:limit_mu}}
\label{sec:proof_limit_mu}

\begin{proof}[Proof of Theorem \ref{theorem:limit_mu}]
Define a corresponding homogeneous finite system $\Lhom_R:=\Lhom\cap B_R$,
which has $N_{\#,R}$ nuclei and $N_{\#,R}$ electrons. Denoting by
$\Nu_{\ell}^{\Omega}(Du(\ell),\tau):=\Nu_{\ell}(Du|_{\Omega}(\ell),\tau)$
for the finite system contained in $\Omega$,
we have
\begin{eqnarray}\label{proof:mu_limit_1}
\nonumber
&&
N(y_R,\mu_R)-N(y_R,\muhom)
\\ \nonumber
&=& \NeR - N_{\#,R}
+ \sum_{\ell\in\Lhom_R}\Nu_{\#}(\pmb{0},\muhom)
- \sum_{\ell\in\L_R}\Nu^{\L_R}_{\ell}(Du(\ell),\muhom)
\\ \nonumber
&=& \big(\NeR - N_{\#,R}\big)
+ \Big( \sum_{\ell\in\Lhom\cap B_{\Rcore}}
\Nu_{\#}\big(\pmb{0},\muhom\big)
- \sum_{\ell\in\L\cap B_{\Rcore}}\Nu^{\L_R}_{\ell}
\big(Du(\ell),\muhom\big) \Big)
\\ \nonumber
&& \quad + \sum_{\ell\in\Lhom_R\backslash B_{\Rcore}}
\Big(\Nu_{\#}(\pmb{0},\muhom)
- \Nu^{\Lhom_R\backslash B_{\Rcore}}_{\ell}(\pmb{0},\muhom) \Big)
\\ \nonumber
&& \quad - \sum_{\ell\in\L_R\backslash B_{\Rcore}}
\Big(\Nu^{\L_R}_{\ell}(Du(\ell),\muhom)
- \Nu^{\L_R\backslash B_{\Rcore}}_{\ell}(Du(\ell),\muhom) \Big)
\\ \nonumber
&& \quad + \sum_{\ell\in\L_R\backslash B_{\Rcore}}
\Big( \Nu^{\L_R\backslash B_{\Rcore}}_{\ell}(\pmb{0},\muhom)
- \Nu^{\L_R\backslash B_{\Rcore}}_{\ell}(Du(\ell),\muhom) \Big)
\\
&=:& T_1+T_2+T_3+T_4+T_5.
\end{eqnarray}
The condition $|\NeR-N_R|<C$ and \asRC~imply that $T_1$ and $T_2$ are uniformly bounded.
$T_3$ can be estimated by Lemma \ref{lemma:limit_fixed_mu} as
\begin{eqnarray}\label{proof:mu_limit_2}
|T_3| \leq C \sum_{\ell\in\L,~\Rcore\leq |\ell|\leq R}
\Big( e^{-\eta_0(|\ell|-\Rcore)} + e^{-\eta_0(R-|\ell|)} \Big)
\leq C R^{d-1}.
\end{eqnarray}
The term $e^{-\eta_0(R-|\ell|)}$ arises due to the presence of the
domain boundary where the local geometry is distinct from the bulk geometry.
 $T_4$ can be bounded in the same way.
To estimate $T_5$, we have from Lemma \ref{lemma:locality_fixed_mu}
that
\begin{eqnarray}\label{proof:mu_limit_3}
\nonumber
|T_5| &\leq& C\sum_{\ell\in\L_R\backslash B_{\Rcore}}
\Big| \Nu^{\L_R\backslash B_{\Rcore}}_{\ell}(\pmb{0},\muhom)
- \Nu^{\L_R\backslash B_{\Rcore}}_{\ell}(Du(\ell),\muhom) \Big|
\\ \nonumber
&\leq&  C \sum_{\ell\in\L_R\backslash B_{\Rcore}} \sum_{\rho\in \L_R\backslash B_{\Rcore}-\ell}
\left| \frac{\partial \Nu^{\L_R\backslash B_{\Rcore}}_{\ell}(Dw(\ell),\muhom)}{\partial D_{\rho}w(\ell)}
\Big|_{Dw = \theta_{\ell} Du} \right| \cdot |D_{\rho}u(\ell)|
\\
&\leq& C\sum_{\ell\in\L\backslash B_{\Rcore}}\big|Du(\ell)\big|_{\gamma}
\leq CN_R^{1/2}\|Du\|_{\ell^2_{\gamma}}
\leq CR^{d/2}\|Du\|_{\ell^2_{\gamma}},
\qquad
\end{eqnarray}
where $\theta_{\ell} \in(0,1)$ depends on $\ell$, and the constant $C$ depends on $\gamma$ and $\gamma_1$.
Therefore, we have from \eqref{proof:mu_limit_1}, \eqref{proof:mu_limit_2}
and \eqref{proof:mu_limit_3} that
\begin{eqnarray}\label{proof:mu_limit_4}
\nonumber
\big| \Nu(u_R,\mu_R)-\Nu(u_R,\muhom) \big|
&\leq& C \big( R^{d-1} + \|Du\|_{\ell^1_{\gamma}} \big)
\\
&\leq& C \big( R^{d-1} + R^{d/2}\|Du\|_{\ell^2_{\gamma}} \big)  .
\end{eqnarray}

Note that for a finite temperature $T>0$,
there exists a constant $c$ depending only on $T$ such that
$f'(\lambda_s-\tau) \geq c~\forall~\lambda_s
\in [\underline{\lambda},\bar{\lambda}]$,
hence
\begin{eqnarray}\label{proof:limit_mu_d_N}
\frac{\partial\Nu(u_R,\tau)}{\partial\tau}
= -\sum_{s=1}^{N_R} f'(\lambda_s-\tau) \geq cN_R \geq CR^d .
\end{eqnarray}
Since \eqref{proof:mu_limit_4} is equivalent to
\begin{eqnarray*}
\bigg| \frac{\partial\Nu(u_R,\tau_{\theta})}{\partial\tau} \bigg| \cdot |\mu_R-\mu_{\#}|
\leq C \big( R^{d-1} + \|Du\|_{\ell^1_{\gamma}} \big)
\leq C \big( R^{d-1} + R^{d/2}\|Du\|_{\ell^2_{\gamma}} \big)
\end{eqnarray*}
%
with $\tau_{\theta}=\theta\mu_R+(1-\theta)\mu_{\#}$, which together with \eqref{proof:limit_mu_d_N}
yields \eqref{eq:mu_converge}
and thus completes the proof.
\end{proof}

\subsection{Proof of Theorem \ref{theorem:limit_problem_1}}
\label{sec:proof_theo_1}

We will first need the following result, which gives us the far-field structure
of the minimizers of \eqref{problem_min_inf}.
For the proof we refer to \cite{chenpre_vardef,ehrlacher13}.

\begin{lemma}\label{lemma:regularity}
	If $\bar{u}\in\Admu(\L)$ is a strongly stable solution to
	\eqref{problem_min_inf} in the sense of
	\eqref{ubar_stability} with some constant $\bar{c}>0$,
	then there exists a constant $C>0$ such that
	\begin{equation}\label{eq:decay_Du}
	|D\bar{u}(\ell)|_{\gamma} \leq C(1+|\ell|)^{-d} \qquad \forall \ell \in \L.
	\end{equation}
\end{lemma}

Next, we shall derive a force-balance equation that is equivalent to
\eqref{problem_min_y} and \eqref{problem_min_y_tau}.
Let $y$ be a configuration with corresponding chemical potential $\mu=\mu(y)$ satisfying \eqref{eq:def_mu}.
Then by using \eqref{eq:Helmoltz_FE_eigen}, \eqref{operator-trace}, \eqref{eq:O_contour} and the fact $\partial_x \mathfrak{g}(x,\tau)=2f(x-\tau)$,
we can compute the derivative of $E(y)$ on the $\ell$-th atom
(ignoring the Cartesian coordinate for simplicity of notations)
\begin{eqnarray}\label{proof:force_HF}
\nonumber
\frac{\partial E(y)}{\partial y(\ell)}
&=& \frac{\partial}{\partial y(\ell)} {\rm Tr} \big( \mathfrak{e}(\mathcal{H}(y),\mu(y)) \big)
\\ \nonumber
&=& \frac{\partial}{\partial y(\ell)} \Big( \mu(y) {\rm  Tr}  \big( \mathfrak{n}(\mathcal{H}(y),\mu(y)) \big) \Big)
+  \frac{\partial}{\partial y(\ell)}{\rm  Tr}  \big( \mathfrak{g}(\mathcal{H}(y),\mu(y)) \big)
\\ 
&=& 2\sum_{s=1}^{\Nn} f(\lambda_s-\mu) \bigg\<
\psi_s \bigg| \frac{\partial\mathcal{H}(y)}{\partial y(\ell)}
\bigg| \psi_s \bigg\> ,
\end{eqnarray}
%
%
To see the last equality of \eqref{proof:force_HF}, we have from \eqref{eq:def_mu} and \eqref{eq:O_contour} that
\begin{eqnarray}\label{proof:force_HF_arXiv}
\nonumber
\frac{\partial E(y)}{\partial y(\ell)}
&=& \frac{\partial}{\partial y(\ell)} {\rm Tr} \big( \mathfrak{g}(\mathcal{H}(y),\mu) \big)
\qquad{\rm with~fixed~}\mu
\\ \nonumber
&=& \frac{1}{2\pi i} \oint_{\mathscr{C}} \mathfrak{g}(z,\mu)
{\rm Tr}\Big(\rzz \mathcal{H}_{,\ell}(y)\rzz\Big)\dd z
\\ \nonumber
&=&  \frac{1}{2\pi i} \sum_{s=1}^{\Nn}\sum_{t=1}^{\Nn}
\oint_{\mathscr{C}}  \frac{\mathfrak{g}(z,\mu)}{(z-\lambda_s)(z-\lambda_t)}  \dd z
\bigg\< \psi_s \bigg| \frac{\partial\mathcal{H}(y)}{\partial y(\ell)} \bigg| \psi_t \bigg\>
\big\< \psi_t \big| \psi_s \big\>
\\ \nonumber
&=&  \frac{1}{2\pi i} \sum_{s=1}^{\Nn} \oint_{\mathscr{C}}  \frac{\mathfrak{g}(z,\mu)}{(z-\lambda_s)^2}  \dd z
\bigg\< \psi_s \bigg| \frac{\partial\mathcal{H}(y)}{\partial y(\ell)} \bigg| \psi_s \bigg\>
\\ 
&=& 2\sum_{s=1}^{\Nn} f(\lambda_s-\mu) 
\bigg\< \psi_s \bigg| \frac{\partial\mathcal{H}(y)}{\partial y(\ell)} \bigg| \psi_s \bigg\> .
\end{eqnarray}

Therefore, any minimiser $(\bar{y},\bar{\mu})$ of \eqref{problem_min_y} and \eqref{problem_min_y_tau}
satisfies the equations
\begin{equation} \label{problem_force_blance_yN}
\left\{\begin{array}{rl}
\Fy_{\ell}(\bar{y},\bar{\mu}) \!\!&= 0 \qquad\forall~\ell\in\Omega\backslash\Omega^{\rm D} ,
\\[1ex]
\Ny(\bar{y},\bar{\mu}) \!\!&= \Ne
\end{array}\right.
\end{equation}
with
\begin{eqnarray}\label{eq:force_tau}
\Fy_{\ell}(y,\tau) := -2\sum_{s=1}^{\Nn} f(\lambda_s-\tau)
\bigg\< \psi_s \bigg| \frac{\partial\mathcal{H}(y)}{\partial y(\ell)}
\bigg| \psi_s \bigg\> .
\end{eqnarray}
We remark that \eqref{problem_force_blance_yN} can also be derived as the
associated Euler-Lagrange equation of the constrained minimization problem
\eqref{problem_min_y_tau}.

We now consider the sequence of problems \eqref{equil_force_R} with parameter $R$.
For $(u_R,\tau)\in \WR(\L) \times\R$, we define
\begin{eqnarray*}
\Fu^{\L_R}_{\ell}(u_R,\tau):=F_{\ell}((x+u_R)|_{{\L_R}}, \tau)
\end{eqnarray*}
and $\Fu(u_R,\tau)\in \WR(\L)'$ with
\begin{eqnarray}\label{eq:F_LR}
\big\< \Fu^{\L_R}(u_R,\tau) , v\big\> = \sum_{\ell\in\L_R}
\Fu^{\L_R}_{\ell}(u_R,\tau) \cdot v(\ell)
\qquad\forall~v\in \WR(\L) .
\end{eqnarray}
Here $\WR(\L)'$ is the dual space of $\WR(\L)$. 


Define $\T_R:  \WR(\L) \times\R\to \WR(\L)'\times\R$ by (recall the definition
of $\Nu^{\L_R}$ below \eqref{equil_force_R})
%
\begin{eqnarray*}
	\T_R(u_R,\tau) := \Big( -\Fu^{\L_R}(u_R,\tau) ,
	\NeR^{-1}\Nu^{\L_R}(u_R,\tau)-1 \Big)
	\quad{\rm for}~(u_R,\tau)\in \WR(\L) \times\R.
\end{eqnarray*}
We have from \eqref{problem_force_blance_yN} that
\eqref{equil_force_R} is equivalent to
\begin{eqnarray}\label{equil_T}
\T_R(\bar{u}_R,\bar{\mu}_R)=\pmb{0}.
\end{eqnarray}
We can further compute the Jacobian matrix of $\T$ at $(u_R,\tau)\in \WR(\L) \times\R$,
\begin{eqnarray}\label{Jac_T}
\J_R(u_R,\tau):=\left[\begin{array}{cc}
-\delta_u \Fu^{\L_R}(u_R,\tau) &
-\delta_{\tau}\Fu^{\L_R}(u_R,\tau)
\\
\NeR^{-1}\delta_u \Nu^{\L_R}(u_R,\tau) &
\NeR^{-1}\delta_{\tau} \Nu^{\L_R}(u_R,\tau)
\end{array}\right] ,
\end{eqnarray}
which will be heavily used in the proof.

\begin{proof}[Proof of Theorem \ref{theorem:limit_problem_1}]
	{\it Step 1. Quasi-best approximation.}
	Following \cite[Lemma 7.3]{ehrlacher13},
	we can construct $T_R\bar{u}\in \AdmR$ such that for $R$ sufficiently large,
	\begin{eqnarray*}\label{proof-4-3-1}
		\|DT_R\bar{u}-D\bar{u}\|_{\ell^2_\gamma}
		\leq C \|D\bar{u}\|_{\ell^2_\gamma(\Lambda\backslash B_{R/2})}
		\leq CR^{-d/2}
	\end{eqnarray*}
	where Lemma \ref{lemma:regularity} is used for the last inequality.
	We now fix some $r > 0$ such that $x+B_r(\bar{u})\subset\Adm_{\mathfrak{m}}$
	for some $\mathfrak{m}>0$. Then, for $R$ sufficiently large,
	we have that $T_R \bar{u} \in B_{r/2}(\bar{u})$ and hence
	$x + B_{r/2}(T_R \bar{u})\subset\Adm_{\mathfrak{m}}$.

	Since $\Gu\in C^3(\AdmR)$, $\delta\Gu$ and $\delta^2\Gu$ are Lipschitz
	continuous in $B_r(\bar{u})\cap\AdmR$
	with Lipschitz constants $L_1$ and $L_2$, that is,
	\begin{align}
	\label{proof:appr_L1}
	\|\delta\Gu(\bar{u})-\delta\Gu(T_R\bar{u})\|
	&\leq L_1\|D\bar{u}-DT_R(\bar{u})\|_{\ell^2_\gamma} \leq CR^{-d/2},
	\qquad \text{and} \\
	\label{proof:appr_L2}
	\|\delta^2\Gu(\bar{u})-\delta^2\Gu(T_R\bar{u})\|
	&\leq L_2\|D\bar{u}-DT_R(\bar{u})\|_{\ell^2_\gamma} \leq CR^{-d/2}.
	\end{align}

	{\it Step 2. Consistency.}
	Let $\Gu^{\L_R}(u,\tau) := \Gy\big((x+u)|_{\L_R},\tau\big)$.
	We have from \eqref{force:G} and \eqref{eq:force_tau} that
	\begin{eqnarray}\label{proof:consistency_f}
	\Fu^{\L_R}_{\ell}(T_R\bar{u},\muhom) = -\frac{\partial\Gu^{\L_R}(T_R\bar{u},\muhom)}{\partial T_R\bar{u}(\ell)} ,
	\end{eqnarray}
	which implies that for any $v\in \WR(\L)$,
	\begin{multline}\label{proof:consistency_a}
	\big\< -\Fu^{\L_R}(T_R\bar{u},\muhom),v \big>
	= -\sum_{\ell\in\L\cap B_{R+R_0}}\Fu^{\L_R}_{\ell}
	(T_R\bar{u},\muhom) v_{\ell}
	\\
	= \big\< \delta_u\Gu^{\L_R}(T_R\bar{u},\muhom),v \big>
	= \sum_{\ell\in \L_R}
	\big\< \delta_{u}\Gu^{\L_R}_{\ell}(D(T_R\bar{u})(\ell), \muhom),Dv(\ell) \big\> .
	\end{multline}
	Using Lemma \ref{lemma:limit_fixed_mu} and the facts that $v=0$
	and $T_R\bar{u} = 0$ outside $\L_R$,
	we have that there exists a constant $\gamma_{\rm c}$, such that
	\begin{equation}\label{proof:consistency_b}
	\left| \big\< \delta_u\Gu^{\L_R}(T_R\bar{u},\muhom)
	- \delta_u\Gu(T_R\bar{u},\muhom) ,v \big> \right|
	\leq Ce^{-\gamma_{\rm c}\Rb} R^{d-1/2}\|Dv\|_{\ell^2_\gamma}.
	\end{equation}
	The proof of this estimate is relatively straightforward and we refer
	to \cite[Proof of (4.12)]{chen15a} for an analogous one.
	In order to balance the error, we must choose $\Rb$ such that
	$e^{-\gamma_{\rm c} \Rb} R^{d-1/2} \leq C R^{-d/2}$, or
	equivalently, $e^{-\gamma_{\rm c} \Rb} \leq C R^{-(3d+1)/2}$.
	On taking logarithms, we observe that this is true provided that
	$\Rb\geq c_{\rm b}\log R$ for $c_{\rm b}$ sufficiently large.

	Then we obtain from \eqref{proof:appr_L1}, \eqref{proof:consistency_a},
	\eqref{proof:consistency_b} and $\delta\Gu(\bar{u}) = 0$
	that $\forall~v\in \WR(\L)$,
	\begin{eqnarray}\label{proof:consistency_c}
	\nonumber
	&& \big\< -\Fu^{\L_R}(T_R\bar{u},\muhom),v \big>
	\\ \nonumber
	&=& \big\<\delta_u\Gu(T_R\bar{u},\muhom)
	- \delta_u\Gu^{\L_R}(T_R\bar{u},\muhom),v\big\> +
	\big\< \delta_{u}\Gu(\bar{u},\muhom)-\delta_u\Gu(T_R\bar{u},\muhom) , v \big\>
	\\
	&\leq& C\big(e^{-\gamma_{\rm c}R} R^{d-1/2}
	+ R^{-d/2}  \big)\|Dv\|_{\ell^2_\gamma}
	~\leq~ CR^{-d/2}\|Dv\|_{\ell^2_\gamma}
	\end{eqnarray}
	for sufficiently large $R$ and appropriate $c_{\rm b}$.

	To proceed, we recall from \eqref{eq:local-homogeneous-observable} the
	definition of local analytic QoIs $\Au_\#$ (in particular $\Nu_\#$) associated
	with the homogeneous lattice.

	We still have to estimate the residual of
	$\NeR^{-1}\Nu^{\L_R}(T_R\bar{u},\muhom)-1$.
	We first construct a corresponding homogeneous finite system
	$\Lhom\cap B_{R+R_{\rm b}}$ with $N_{\#,\L_R}$ electrons,
	and then obtain from an argument similar to
	\eqref{proof:mu_limit_1}-\eqref{proof:mu_limit_4} that
	\begin{eqnarray*}
		&& \big|\Nu^{\L_R}(T_R\bar{u},\muhom)-\NeR\big|
		\\
		&\leq& \Big| \Nu^{\L_R}(T_R\bar{u},\muhom)
		- \sum_{\ell\in\Lhom \cap B_{R+R_{\rm b}}} \Nu_{\#}(\pmb{0},\muhom) \Big|
		+ \big|N_{{\rm e},R} - N_{\#,\L_R}\big|
		\\
		&\leq & C\big(R^{d-1}+R^{d/2}\big),
	\end{eqnarray*}
	where $C$ depends on $\|DT_R \bar{u} \|_{\ell^2_\gamma}$
	(which is bounded by $\|D \bar{u} \|_{\ell^2_\gamma}$ and is hence
	independent of $R$).
	%
	Therefore, we have
	\begin{eqnarray}\label{proof:consistency_d}
	\big|\NeR^{-1}\Nu^{\L_R}(T_R\bar{u},\muhom)-1\big|
	\leq CR^{-\min\{1,d/2\}} .
	\end{eqnarray}

	Then we have the following consistency estimate from
	\eqref{proof:consistency_c} and \eqref{proof:consistency_d}
	\begin{eqnarray}\label{proof:consistency}
	\left\|\T_R\big(T_R\bar{u},\muhom\big)\right\|_{\WR(\L)'\times\R}
	\leq CR^{-\min\{1,d/2\}} .
	\end{eqnarray}

	{\it Step 3. Stability.}
	We have from \eqref{proof:consistency_f} that
	\begin{eqnarray}\label{proof:stability_a}
	\delta_u \Fu^{\L_R}(T_R\bar{u},\muhom) =
	-\delta_u^2\Gu^{\L_R}(T_R\bar{u},\muhom).
	\end{eqnarray}
	Using Lemma \ref{lemma:limit_fixed_mu} and the facts that $v=0$
	and $T_R\bar{u} = 0$ outside $B_R$,
	we have that there exists a constant $\gamma_{\rm s}$, such that
	\begin{equation}\label{proof:stability_b}
	\left| \big\< \big(\delta^2_u\Gu^{\L_R}(T_R\bar{u},\muhom)
	- \delta^2_u\Gu(T_R\bar{u},\muhom) \big) v ,v \big> \right|
	\leq Ce^{-\gamma_{\rm s}\Rb} R^{d}\|Dv\|_{\ell^2_\gamma}^2.
	\end{equation}
	The proof of this estimate is straightforward and we refer to
	\cite[Proof of (4.10)]{chen15a} for an analogous one.
	Together with \eqref{ubar_stability} and \eqref{proof:appr_L2} this
	leads to
	\begin{eqnarray}\label{proof:stability_c}
	\nonumber
	&& \big\< -\delta_u \Fu^{\L_R}(T_R\bar{u},\muhom) v,v \big\>
	~=~ \big\<\delta^2\Gu^{\L_R}(T_R\bar{u},\muhom)v,v\big\>
	\\ \nonumber
	&=& \big\<\delta^2\Gu(\bar{u},\muhom)v,v\big\>
	+ \big\<\big(\delta^2\Gu(T_R\bar{u},\muhom)
	-\delta^2\Gu(\bar{u},\muhom)\big)v,v\big\>
	\\ \nonumber
	&& + \big\<\big(\delta^2\Gu^{\L_R}(T_R\bar{u},\muhom)
	-\delta^2\Gu(T_R\bar{u},\muhom)\big)v,v\big\>
	\\
	&\geq& \big( \bar{c} - C( R^{-d/2} + e^{-\gamma_{\rm s}\Rb}R^{d}) \big)
	\|Dv\|_{\ell^2_\gamma}^2
	~\geq~ \frac{\bar{c}}{2}\|Dv\|_{\ell^2_\gamma}^2
	\qquad\forall~v \in \WR(\L) \qquad
	\end{eqnarray}
	for sufficiently large $R$ and $c_{\rm b}$. Therefore,
	$-\delta_u \Fu^{\L_R}(T_R\bar{u},\muhom)$ is positive definite.

	Similar to \eqref{proof:limit_mu_d_N}, we have that there exists
	a constant $a_0>0$ such that
	\begin{eqnarray}\label{proof:stability_d}
	\NeR^{-1}\delta_{\tau} \Nu^{\L_R}(T_R\bar{u},\muhom) \geq a_0.
	\end{eqnarray}
	It remains to estimate the off-diagonal terms of $\J_R(T_R\bar{u},\muhom)$.
	Note that
	\begin{eqnarray*}
		\frac{\partial\Ny^{\L_R}(y,\tau)}{\partial y(\ell)}
		= 2\sum_{s=1}^{N_R} f'(\lambda_s-\tau) \bigg\<\psi_s \left|
		\frac{\partial\mathcal{H}(y^{\L_R})}{\partial y(\ell)}
		\right| \psi_s\bigg\>,
	\end{eqnarray*}
	which together with
	\begin{eqnarray*}
		\frac{-\partial\Fy^{\L_R}_{\ell}(y,\tau)}{\partial\tau}
		= 2\sum_{s=1}^{N_R} f'(\lambda_s-\tau)\
		\bigg\<\psi_s \left|
		\frac{\partial\mathcal{H}(y^{\L_R})}{\partial y(\ell)}
		\right| \psi_s\bigg\>
	\end{eqnarray*}
	implies
	\begin{eqnarray}\label{proof:stability_h}
	-\delta_{\tau}\Fu^{\L_R}(T_R\bar{u},\muhom)
	= \delta_u \Nu^{\L_R}(T_R\bar{u},\muhom) .
	\end{eqnarray}
	We then observe that there exists a constant $b_0>0$ such hat
	\begin{equation}\label{proof:stability_e}
	\big\< \delta_u\Nu^{\L_R}(T_R\bar{u}, \muhom), v \>
	\leq b_0 \|Dv\|_{\ell^2_{\gamma}}
	\qquad\forall~v\in \WR(\L).
	\end{equation}
	To see this, we have
	\begin{multline*}
	\big\< \delta_u\Nu^{\L_R}(T_R\bar{u}, \muhom), v \>
	\\ = \big\< \delta_u\Nu^{\L_R}(T_R\bar{u}, \muhom)
	- \delta_u\Nu(T_R\bar{u}, \muhom), v \>
	+ \big\< \delta_u\Nu(T_R\bar{u}, \muhom), v \>, \qquad
	\end{multline*}
	where the first term is estimated analogously to
	\eqref{proof:consistency_b}
	\begin{equation*}
	\left| \big\< \delta_u\Nu^{\L_R}(T_R\bar{u},\muhom)
	- \delta_u\Nu(T_R\bar{u},\muhom) ,v \big> \right|
	\leq Ce^{-\gamma_{\rm c}\Rb} R^{d-1/2}\|Dv\|_{\ell^2_\gamma}
	\end{equation*}
	and the second term can be estimated by using similar arguments
	as those in \cite{chenpre_vardef} and \cite[Lemma 2.1]{ehrlacher13}
	\begin{eqnarray*}
	\big\< \delta_u\Nu(T_R\bar{u}, \muhom), v \big\>
	&=& \sum_{\ell\in \L} \big\< \delta_u\Nu_{\ell}(\pmb{0} , \muhom),Dv(\ell) \big\>
	- \sum_{\ell\in \Lhom} \big\< \delta_u\Nu_{\#}(\pmb{0}, \muhom), Dv(\ell) \big\>
	\\
	&& + \sum_{\ell\in\L} \big\< \delta^2_u\Nu_{\ell}(\theta_{\ell} T_R\bar{u},
	\muhom) Du(\ell), Dv(\ell) \big\>
	\\
	&\leq& C\|Dv\|_{\ell^2_{\gamma}}
	\qquad
	\end{eqnarray*}
	with $\theta_{\ell}\in(0,1)$ depending on $\ell$ and the constant $C$ depending on $T_R\bar{u}$.

	To show the stability,  we want to solve
	\begin{displaymath}
	\J_R\big(T_R\bar{u},\muhom\big) (v, \kappa) = (w, \xi)
	\end{displaymath}
	for any $(w,\xi)\in \WR(\L)'\times\R$.
	We first obtain from \eqref{proof:stability_d} and \eqref{proof:stability_e}
	that
	\begin{eqnarray}\label{proof:stability_f}
	|\kappa| \leq a_0^{-1}\Big|\xi-\NeR^{-1}b_0\|Dv\|_{\ell^2_{\gamma}}\Big|
	\leq C|\xi|.
	\end{eqnarray}
	To obtain $v$, we solve
	\begin{displaymath}
	\delta_u^2 \Gu^{\L_R}(T_R\bar{u},\muhom) v
	= w + \kappa\delta_u\Gu^{\L_R}(T_R\bar{u},\muhom).
	\end{displaymath}
	By using \eqref{proof:stability_c}, \eqref{proof:stability_h}
	and \eqref{proof:stability_e}, we obtain that
	\begin{eqnarray*}\label{proof:stability_g}
		\|Dv\|_{\ell^2_{\gamma}} \leq C\big(\|w\|_{\WR(\L)'} + b_0|\xi|\big)
		\leq C\big(\|w\|_{\WR(\L)'} + |\xi|\big).
	\end{eqnarray*}
	This together with \eqref{proof:stability_f} implies that
	\begin{eqnarray}\label{proof:stability}
    \J_R\big(T_R\bar{u},\muhom\big): \WR(\L) \times\R\rightarrow \WR(\L)'\times\R ~\text{ is an isomorphism.}
	\end{eqnarray}

	{\it Step 4. Application of Inverse Function Theorem.}
	With the consistency \eqref{proof:consistency} and the stability \eqref{proof:stability},
	we can apply the inverse function theorem \cite[Lemma B.1]{luskin13}
	on the function $\T_R$ around the point $\big(T_R\bar{u}, \muhom\big)$,
	to obtain the existence of $\bar{u}_R$ and the estimate \eqref{err_Du}.
\end{proof}

\subsection{Proof of Theorem \ref{theorem:limit_problem_2}}
\label{sec:proof_theo_2}

\begin{proof}[Proof of Theorem \ref{theorem:limit_problem_2}]
Since $\| D\bar{u}_{R_j} \|_{\ell^2_\gamma}$ is bounded, we have from
Theorem \ref{theorem:limit_mu} that
\begin{eqnarray}\label{proof_mu_3}
|\bar{\mu}_{R_j}-\muhom| \leq CR^{-\min\{1,d/2\}} .
\end{eqnarray}

Again using $\sup_j\| D\bar{u}_{R_j} \|_{\ell^2_\gamma}<\infty$, we have from
the Banach-Alaoglu theorem (note that $\UsH$ becomes a Hilbert space after
factoring out a constant shift) that there exists a subsequence (not relabelled)
and $\bar{u}\in\AdmR$ such that
\begin{eqnarray*}
\bar{u}_{R_j} \stackrel{w~}{\rightharpoonup}\bar{u} \qquad{\rm in}~\UsH(\L)
\quad{\rm as}~ j\rightarrow\infty.
\end{eqnarray*}
Since $v\mapsto D_{\rho}v(\ell)$ is a linear functional on $\UsH(\L)$
for any $\ell\in\L$ and $\rho\in\L-\ell$, we have \eqref{convergence_uj}.

Now it is only necessary to show that $\del\Gu(\bar{u})=0$. Let $v \in \UsH(\L)$
have compact support $B_{R_v}$. Then, for $j$ sufficiently large,
$v|_{\L_{R_j}}$ is an admissible test function for \eqref{equil_force_R}.
Thus, using \eqref{nablaG_nablaE} we obtain
\begin{eqnarray*}
	0 = \big\< \del_u \Gu^{\L_{R_j}}(\bar{u}_{R_j}, \bar{\mu}_{R_j}), v \big\>
	  = \sum_{\ell \in \Omega_v}
	  \Fu_\ell^{\L_{R_j}}(\bar{u}_{R_j}, \bar{\mu}_{R_j}) \cdot v(\ell) ,
\end{eqnarray*}
where $\Omega_v:=\big\{\ell\in\L, ~v(\ell)\neq 0\big\}$.
To complete the proof we only need to show that
\begin{equation} \label{eq:weak-conv-thm-conv-frc}
	\Fu_\ell^{\L_{R_j}}(\bar{u}_{R_j}, \bar{\mu}_{R_j})
	\to \Fu_\ell(\bar{u}, \muhom)
	\qquad \text{for all $\ell \in \Omega_v$, as $j \to \infty$.}
\end{equation}
To see this, we have
\begin{eqnarray*}
&& \big| \Fu_\ell^{\L_{R_j}}(\bar{u}_{R_j}, \bar{\mu}_{R_j})
- \Fu_\ell(\bar{u}, \muhom) \big|
\\
&=& \big| \Fu_\ell^{\L_{R_j}}(\bar{u}_{R_j}, \bar{\mu}_{R_j})
- \Fu^{\L_{R_j}}_\ell(\bar{u}, \bar{\mu}_{R_j}) \big|
+ \big| \Fu^{\L_{R_j}}_\ell(\bar{u}, \bar{\mu}_{R_j})
- \Fu_\ell(\bar{u}, \bar{\mu}_{R_j}) \big|
\\
&& + \big| \Fu_\ell(\bar{u}, \bar{\mu}_{R_j}) - \Fu_\ell(\bar{u}, \muhom) \big|
\\
&=:& T_{1,j} + T_{2,j} + T_{3,j}.
\end{eqnarray*}
Note that \eqref{convergence_uj} implies
\begin{eqnarray*}
\limsup_{j\rightarrow\infty}
\|D\bar{u}-D\bar{u}_{R_j}\|_{\ell^2_{\gamma}(\L\cap B_{2R_v})} = 0 ,
\end{eqnarray*}
which together with the fact
\begin{equation}\label{eq:F_Dul}
\Fu_{\ell}^\Omega(u,\tau) =
\sum_{\rho\in \ell-\Omega} \Gu^{\Omega}_{\ell-\rho,\rho}\big(Du(\ell-\rho), \tau\big)
- \sum_{\rho\in\Omega-\ell} \Gu^{\Omega}_{\ell,\rho}\big(Du(\ell), \tau\big)
\qquad{\rm with}~\Omega=\L_{R_j}
\end{equation}
and Lemma \ref{lemma:locality_fixed_mu} leads to
\begin{eqnarray}\label{proof:LRj1}
\limsup_{j\rightarrow\infty}|T_{1,j}|\leq Ce^{-\eta_1 R_v} .
\end{eqnarray}
Using \eqref{eq:locality_FE_limit} and \eqref{eq:F_Dul} with $\Omega=\L$,
we can estimate $T_{2,j}$ by
\begin{eqnarray}\label{proof:LRj2}
|T_{2,j}|\leq Ce^{-\eta_1 R_j} .
\end{eqnarray}
Finally, we have from \eqref{proof_mu_3} that
\begin{eqnarray}\label{proof:LRj3}
	|T_{3,j}| \leq C|\bar{\mu}_{R_j}-\muhom| \leq CR_j^{-\min\{1,d/2\}} .
\end{eqnarray}
Taking into accounts \eqref{proof:LRj1}, \eqref{proof:LRj2}, \eqref{proof:LRj3}
and the fact that $R_v$ can be chosen arbitrarily large (independent of $j$),
we obtain \eqref{eq:weak-conv-thm-conv-frc} and complete the proof.
\end{proof}

\appendix
\renewcommand\thesection{\appendixname~\Alph{section}}

\section{Periodic boundary conditions}
\label{sec:pbc}
\renewcommand{\theequation}{A.\arabic{equation}}
\renewcommand{\thetheorem}{A.\arabic{theorem}}
\renewcommand{\thelemma}{A.\arabic{lemma}}
\renewcommand{\theproposition}{A.\arabic{proposition}}
\renewcommand{\thealgorithm}{A.\arabic{algorithm}}
\renewcommand{\theremark}{A.\arabic{remark}}
\setcounter{equation}{0}

Periodic boundary condition (i.e. the {\it supercell model}) is the most popular
choice for simulating crystalline defects. To implement periodic boundary
conditions, let $\Omega_R\subset\R^d$ be connected such that
$B_R\subset\Omega_R$, for ${\sf B}=(b_1,\cdots,b_d)\in\R^{d\times d}$ non-singular,
$b_i\in\Lhom$, $\bigcup_{\alpha\in\Z^d}\{{\sf
B}\alpha+\Omega_R\} = \R^d$, and the shifted domains ${\sf B}\alpha+\Omega_R$
are disjoint. The computational cell is defined by
\begin{eqnarray*}
\Lp_R:=\Omega_R \cap \L.
\end{eqnarray*}

We consider a {\it torus} tight binding model, defined as follows;
an alternative periodic model is desribed in Remark~\ref{rem:bzint}.
Let $N_R:=\#(\Lp_R)$.
Then for $y:\Lp_R\rightarrow\R^d$, the Hamiltonian matrix $\mathcal{H}^{\#}(y)\in\R^{N_R\times N_R}$ has the matrix elements
\begin{eqnarray}\label{tb-H-elements-torus}
\Big(\mathcal{H}^{\#}(y)\Big)_{\ell k}
= \left\{ \begin{array}{ll}
\displaystyle
h_{\rm ons}\bigg(
\sum_{\substack{j \in \Lp_R \setminus \ell \\ \alpha \in \Z^d}}
		\varrho(r_{\ell (j+{\sf B} \alpha)})  \bigg)
+ \sum_{\alpha\in\Z^d\backslash\pmb{0}} h_{\rm hop}(r_{\ell (\ell+{\sf B}\alpha)})
& {\rm if}~\ell=k; \\[1ex]
\displaystyle
\sum_{\alpha\in\Z^d} h_{\rm hop}(r_{\ell (k+{\sf B}\alpha)}) & {\rm if}~\ell\neq k,
\end{array} \right.
\end{eqnarray}
where $\varrho$, $h_{\rm ons}$ and $h_{\rm hop}$ are given by \eqref{tb-H-elements},
and $r_{\ell(j+B\alpha)} = |B\alpha + y(j) - y(\ell)|$.
We can then compute the eigenpairs of $\mathcal{H}^{\#}(y)$ and define the corresponding 
(local) analytic QoIs and (local) density of states analogously as in \S~\ref{sec:model_finite}.
We will denote these objects by the same symbolds as in \S~\ref{sec:model_finite}.

Repeating the proofs in \S~\ref{sec:proof_locality} verbatim, we obtain locality result
of local analytic QoIs:
\begin{eqnarray}\label{eq:locality_ldos_A}
\left|\frac{\partial^j \Ay_{\ell}(y,\tau)}
{\partial [y(m_1)]_{i_1}\cdots\partial [y(m_j)]_{i_j}}\right|
\leq C_j e^{-\gamma_j\sum_{t=1}^j r^{\#}_{\ell m_t}},
\end{eqnarray}
which is identical to Lemma \ref{lemma:decay_el}, but the distance $r_{\ell k}$
is replaced with the {\it torus distance}
\begin{eqnarray*}
r^{\#}_{\ell k} := \min_{\alpha\in \Z^d} \big| y(\ell)-y(k)+{\sf B}\alpha \big| .
\end{eqnarray*}
Analogously to \S~\ref{sec:limit_ldos}, we can again define the pointwise
thermodynamic limit of the local density of states and of local analytic QoIs,
and observe that they inherit again the locality \eqref{eq:locality_ldos_A}.

Turning to the formulation of force equilibration, the set of
admissible displacements is now given by
\begin{multline*}
\Admu^{\#}(R) := \big\{ u:\L_R^{\#}\rightarrow\R^d ~\big|~
|y(\ell)-y(k)+{\sf B}\alpha| \geq \mathfrak{m}|\ell-k+{\sf B}\alpha|
\\ ~~{\rm for~any}~\ell,k\in\Lp_R~{\rm and}~\alpha\in\Z^d,~{\rm for~some~}\mathfrak{m}>0 \big\} . \qquad
\end{multline*}
The Helmholtz free energy for $u\in\Admu^{\#}(R)$ is given by
\begin{eqnarray*}
\Eu^{\#}_R(u) = \Eu^{\#}_R\big(u,\mu(u)\big)
:= \sum_{\ell\in\Lp_R}\Eu_{\ell}^{\L_R^{\#}}\big(Du(\ell),\mu\big) ,
\end{eqnarray*}
where the chemical potential $\mu=\mu(u)$ is chosen such that
\begin{eqnarray}\label{def:mu_periodic}
\NeR = \Nu^{\#}_R(u,\mu)
:= \sum_{\ell\in\Lp_R} \Nu_{\ell}^{\L_R^{\#}}\big(Du(\ell),\mu\big)
\end{eqnarray}
with $\NeR$ a prescribed number of electrons contained in $\Omega_R$.
Here, $\Eu_{\ell}^{\L_R^{\#}}$ and $\Nu_{\ell}^{\L_R^{\#}}$ are local
analytic QoIs for the above torus model. 

We can now derive the limit of chemical potential with periodic boundary
conditions, which is an analogous result to Theorem \ref{theorem:limit_mu}, but
with an improved convergence  rate due to the fact that boundary effects no
longer occur.

\begin{theorem}\label{theorem:mu_limit_per}
Let $\L$ satisfy \asRC, $\Lp_R := \L \cap B_R \uparrow \L$ and $N_R := \#(\Lp_R)$.
For each $R$ let $u^{\#}_R : \Lp_R \to \R^d$
with $y_R^{\#}(\ell) := \ell + u_R^{\#}(\ell)$ a configuration with parameter $\mathfrak{m}$ independent of $R$.

Let $\NeR \in \mathbb{R}$ be a prescribed number of electrons in the subsystem
$\L_R$, chosen such that $|N_R - \NeR|$ is bounded as $R \to \infty$. Then, the
chemical potential $\mu_R^{\#}$ solving \eqref{def:mu_periodic} is
well-defined and satisfies
\begin{eqnarray}\label{eq:mu_converge_per}
\big|\mu_R^{\#} - \muhom \big|
\leq C_1^{\#} R^{-d} \| Du_R^{\#} \|_{\ell^1_\gamma}
\leq C_2^{\#} R^{-d/2} \| Du_R^{\#} \|_{\ell^2_\gamma}
\end{eqnarray}
with some constants $C^{\#}_1$ and $C^{\#}_2$.
\end{theorem}

\begin{proof}
	The proof is analogous to that of Theorem \ref{theorem:limit_mu}.
	The main difference lies in that there is no surface term in $T_3$ and $T_4$.
	More specifically, \eqref{proof:mu_limit_2} is replaced by
	\begin{displaymath} 
	|T_3| \leq C \sum_{\ell\in\Lp_R,~|\ell|\geq\Rcore}
					e^{-\gamma_0(|\ell|-\Rcore)}
	\leq C. \qedhere
	\end{displaymath}
\end{proof}

We now consider the thermodynamic limit of the equilibrium problem with
periodic boundary conditions corresponding to \eqref{problem_min_y_tau}:
\begin{eqnarray}\label{equil_force_R_per}
(\bar{u}^{\#}_R,\bar{\mu}^{\#}_R)
\in \arg\min \Big\{ \Eu^{\#}_R(u_R, \tau) ~:~
\Nu^{\#}_R(u_R, \tau) = \NeR, ~u_R\in\Admu^{\#}(R) \Big\} . ~~
\end{eqnarray}

The following two results establish that the thermodynamic limit of
\eqref{equil_force_R_per} as $R\rightarrow\infty$ is again
\eqref{problem_min_inf}, i.e., the same as with clamped boundary conditions. The
proofs are analogous to those of Theorems~\ref{theorem:limit_problem_1} and
\ref{theorem:limit_problem_2}, with the exception of the proof of stability
of the approximation. For the latter we refer to \cite[Thm. 7.7]{ehrlacher13} for
an analogous result that is readily adapted. Hence, we do not give details but
only mention again that the convergence rate is improved here as well.

\begin{theorem}\label{theorem:limit_problem_1_per}
	Assume that $|N_R - \NeR|$ is bounded as $R \to \infty$.
	If $\bar{u}\in\Admu(\L)$ is a strongly stable solution of \eqref{problem_min_inf}
	in the sense that \eqref{ubar_stability} holds with some constant $\bar{c}>0$,
	then, for $R$ sufficiently large there exists a solution
	$(\bar{u}^{\#}_R,\bar{\mu}^{\#}_R)$ of \eqref{equil_force_R_per}
	satisfying
	\begin{eqnarray}\label{err_Du_per}
	\big\|D\bar{u}-D\bar{u}^{\#}_R\big\|_{\ell^2_{\gamma}(\L_R)}
	+ \big|\bar{\mu}^{\#}_R-\muhom\big| \leq CR^{-d/2}.
	\end{eqnarray}
\end{theorem}

\begin{theorem} \label{theorem:limit_problem_2_per}
	Let $R_j \uparrow \infty$ and $(\bar{u}^{\#}_{R_j}, \bar{\mu}^{\#}_{R_j})$
	be solutions to \eqref{equil_force_R_per}.
	If $|N_{R_j} - N_{{\rm e},R_j}|$ is bounded
	and $\sup_{j > 0} \| D\bar{u}_{R_j} \|_{\ell^2_\gamma(\L_R)} < \infty$,
	then there exists a subsequence (not relabelled) and $\bar{u} \in \Admu(\L)$
	such that
	\begin{eqnarray}\label{convergence_uj_per}
	\bar{\mu}^{\#}_{R_j} \to \muhom \quad \text{and} \quad
	D_\rho \bar{u}^{\#}_{R_j}(\ell) \to D_\rho \bar{u}(\ell) \quad
	\forall~\ell \in \L, ~\rho \in \L - \ell.
	\end{eqnarray}
	Moreover, each such accumulation point $\bar{u}$ solves \eqref{problem_min_inf}.
\end{theorem}

\begin{remark} \label{rem:bzint}
	An alternative approach is to approximate the local defect by
	repeating the computational cell periodically, which yields
	an infinite lattice of defects,
	\begin{eqnarray*}
		\L_R^{\rm per}:=\bigcup_{\alpha\in\Z^d}\{{\sf B}\alpha+\Omega_R\}.
	\end{eqnarray*}
The associated set of  admissible displacements is
\begin{multline*}
\Admu^{\rm per}(R) := \big\{ u:\L_R^{\rm per}\rightarrow\R^d ~\big|~
x+u\in\Adm(\L_R^{\rm per}),
\\ ~u(\ell+b_i)=u(\ell)\text{ for }\ell\in\L_R^{\rm per}, ~i=1,\cdots,d \big\} . \qquad
\end{multline*}
The Helmholtz free energy for $u\in\Admu^{\rm per}(R)$ is given by
\begin{eqnarray*}
	\Eu^{\rm per}_R(u) = \Eu^{\rm per}_R\big(u,\mu(u)\big)
	:= \sum_{\ell\in\L_R}\Eu_{\ell}^{\L_R^{\rm per}}\big(Du(\ell),\mu\big) ,
\end{eqnarray*}
where the chemical potential $\mu=\mu(u)$ is chosen such that
\begin{eqnarray}\label{def:mu_periodic_arx}
\NeR = \Nu^{\rm per}_R(u,\mu)
:= \sum_{\ell\in\L_R} \Nu_{\ell}^{\L_R^{\rm per}}\big(Du(\ell),\mu\big)
\end{eqnarray}
with $\NeR$ the number of electrons contained in the periodic cell $\Omega_R$.
Note that the local analytic QoIs $\Eu_{\ell}^{\L_R^{\rm per}}$ and $\Nu_{\ell}^{\L_R^{\rm per}}$ 
are defined by the thermodynamic limits in the infinite lattice of defects $\L_R^{\rm per}$. 
In practise, the quantities
$\Eu^{\rm per}_R(u, \mu)$ and $\NeR(u, \mu)$ are computed via Bloch's theorem
(Brilluoin zone integration).

Then the thermodynamic limit of the equilibrium problem within this setting is given by:
\begin{eqnarray}\label{equil_force_R_per_A}
(\bar{u}^{\rm per}_R,\bar{\mu}^{\rm per}_R)
\in \arg\min \Big\{ \Eu^{\rm per}_R(u_R, \tau) ~:~
\Nu^{\rm per}_R(u_R, \tau) = \NeR, ~u_R\in\Admu^{\rm per}(R) \Big\} . ~~
\end{eqnarray}
We conclude that similar results as Theorem \ref{theorem:mu_limit_per}, \ref{theorem:limit_problem_1_per} and \ref{theorem:limit_problem_2_per} are true within this framework.

\end{remark}


\begin{theorem}\label{theorem:mu_limit_per_arxiv}
	Let $\L$ satisfy \asRC, $\L_R := \L \cap B_R \uparrow \L$ and $N_R := \#
	\L_R$. For each $R$ let $u^{\rm per}_R : \L_R^{\rm per} \to \R^d$
	with $y_R^{\rm per}(\ell) := \ell + u_R^{\rm per}(\ell) \in
	\Adm_{\mathfrak{m}}(\L_R^{\rm per})$, where $\mathfrak{m}$ is independent of $R$.

	Let $\NeR \in \mathbb{R}$ be a prescribed number of electrons in the subsystem
	$\L_R$, chosen such that $|N_R - \NeR|$ is bounded as $R \to \infty$. Then, the
	chemical potential $\mu_R^{\rm per}$ solving \eqref{def:mu_periodic} is
	well-defined and satisfies
	\begin{eqnarray}\label{eq:mu_converge_per_arxiv}
	\big|\mu_R^{\rm per} - \muhom \big|
	\leq C_1^{\rm per} R^{-d} \| Du_R^{\rm per} \|_{\ell^1_\gamma}
	\leq C_2^{\rm per} R^{-d/2} \| Du_R^{\rm per} \|_{\ell^2_\gamma}
	\end{eqnarray}
	with some constants $C^{\rm per}_1$ and $C^{\rm per}_2$.
\end{theorem}

\begin{theorem}\label{theorem:limit_problem_1_per_arxiv}
	Assume that $|N_R - \NeR|$ is bounded as $R \to \infty$.
	If $\bar{u}\in\Admu(\L)$ is a strongly stable solution of \eqref{problem_min_inf}
	in the sense that \eqref{ubar_stability} holds with some constant $\bar{c}>0$,
	then, for $R$ sufficiently large there exists a solution
	$(\bar{u}^{\rm per}_R,\bar{\mu}^{\rm per}_R)$ of \eqref{equil_force_R_per}
	satisfying
	\begin{eqnarray}\label{err_Du_per_arxiv}
	\big\|D\bar{u}-D\bar{u}^{\rm per}_R\big\|_{\ell^2_{\gamma}(\L_R)}
	+ \big|\bar{\mu}^{\rm per}_R-\muhom\big| \leq CR^{-d/2}.
	\end{eqnarray}
\end{theorem}

\begin{theorem} \label{theorem:limit_problem_2_per_arxiv}
	Let $R_j \uparrow \infty$ and $(\bar{u}^{\rm per}_{R_j}, \bar{\mu}^{\rm per}_{R_j})$
	be solutions to \eqref{equil_force_R_per}.
	If $|N_{R_j} - N_{{\rm e},R_j}|$ is bounded
	and $\sup_{j > 0} \| D\bar{u}_{R_j} \|_{\ell^2_\gamma(\L_R)} < \infty$,
	then there exists a subsequence (not relabelled) and $\bar{u} \in \Admu(\L)$
	such that
	\begin{eqnarray}\label{convergence_uj_per_arxiv}
	\bar{\mu}^{\rm per}_{R_j} \to \muhom \quad \text{and} \quad
	D_\rho \bar{u}^{\rm per}_{R_j}(\ell) \to D_\rho \bar{u}(\ell) \quad
	\forall~\ell \in \L, ~\rho \in \L - \ell.
	\end{eqnarray}
	Moreover, each such accumulation point $\bar{u}$ solves \eqref{problem_min_inf}.
\end{theorem}

\section{Dislocations}
\label{sec:dislocation}
\renewcommand{\theequation}{B.\arabic{equation}}
\renewcommand{\thetheorem}{B.\arabic{theorem}}
\renewcommand{\thelemma}{B.\arabic{lemma}}
\renewcommand{\theproposition}{B.\arabic{proposition}}
\renewcommand{\thealgorithm}{B.\arabic{algorithm}}
\renewcommand{\theremark}{B.\arabic{remark}}
\setcounter{equation}{0}

We consider a model for straight dislocation lines obtained by projecting a 3D
crystal. For a 3D lattice $B\Z^3$ with dislocation direction parallel to $e_3$
and Burgers vector $\burg= (\burg_1, \burg_2, \burg_3)=(\burg_1, 0, \burg_3)$, we consider displacements
$W:B\Z^3\rightarrow\R^3$ that are periodic in the direction of the dislocation
direction $e_3$. Thus, we choose a projected reference lattice
$\Lambda:=A\Z^2=\{(\ell_1,\ell_2) ~|~\ell=(\ell_1,\ell_2,\ell_3)\in B\Z^3\}$,
which is again a Bravais lattice.
We can define a macroscopically applied deformation $P \in \R^{2 \times 3}$
by $P(\ell_1,\ell_2)=(\ell_1,\ell_2,\ell_3)$.

Let $\hat{x}\in\R^2$ be the position of the dislocation core and
$\Gamma := \{x \in \R^2~|~x_2=\hat{x}_2,~x_1\geq\hat{x}_1\}$ be the ``branch
cut'', with $\hat{x}$ chosen such that $\Gamma\cap\Lambda=\emptyset$.  Following
\cite{ehrlacher13}, we define the far-field predictor $u_0$ by
\begin{eqnarray}\label{predictor-u_0-dislocation}
u_0(x):=\ulin(\xi^{-1}(x)),
\end{eqnarray}
where $\ulin \in C^\infty(\R^2 \setminus \Gamma; \R^d)$ is the continuum linear
elasticity solution (see \cite{ehrlacher13} for the details) and
\begin{eqnarray}
\xi(x)=x-\burg_{12}\frac{1}{2\pi}
\eta\left(\frac{|x-\hat{x}|}{\hat{r}}\right)
\arg(x-\hat{x}),
\end{eqnarray}
with $\arg(x)$ denoting the angle in $(0,2\pi)$ between $x$ and
$\burg_{12} = (\burg_1, \burg_2) = (\burg_1, 0)$, and
$\eta\in C^{\infty}(\R)$ with $\eta=0$ in $(-\infty,0]$, $\eta=1$ in
$[1,\infty)$ removes the singularity.

The configuration $y$ is now decomposed into
\begin{eqnarray*}
y(\ell)=y_0(\ell)+u(\ell) \qquad \forall~\ell\in\L,
\end{eqnarray*}
where the predictor $y_0 = Px + u_0$ is constructed in such a way that $y_0$
jumps across $\Gamma$ and encodes the presence of the dislocation.
One can treat anti-plane models of pure screw dislocations by admitting
displacements of the form $u_0 = (0, 0, u_{0,3})$ and $u = (0, 0, u_3)$.
Similarly, one can treat the in-plane models of pure edge dislocations by
admitting displacements of the form $u_0 = (u_{0,1}, u_{0,2}, 0)$ and
$u = (u_1, u_2, 0)$ \cite{ehrlacher13}.

There is an ambiguity in the definition of  $y_0$ in that we could have equally
placed the jump into the left half-plane $\{ x_1 \leq \hat{x}_1 \}$. The role
of $\xi$ in the definition of $u_0$ is that applying a plastic slip across the
plane $\{ x_2 = \hat{x}_2 \}$ via the definition
\begin{displaymath}
y^S(x) := \left\{ \begin{array}{ll}
y(\ell), & \ell_2 > \hat{x}_2, \\
y(\ell-\burg_{12}) - \burg_3 e_3, & \ell_2 < \hat{x}_2
\end{array} \right.
\end{displaymath}
achieves exactly this transfer: it leaves the (3D) configuration invariant,
while generating a new predictor $y_0^S \in C^\infty(\Omega_\Gamma)$ where
$\Omega_\Gamma = \{ x_1 > \hat{x}_1 + \hat{r} + \burg_1 \}$. Since the
map $y \mapsto y^S$ represents a relabelling of the atom indices and an
integer shift in the out-of-plane direction, we can apply the isometry and
permutation invariance of $\Ay_{\ell}$ (see Remark \ref{remark:symmetry_limit})
to obtain
\begin{equation}
\label{eq:perm-invariance-site-E-dislocations}
\Ay_\ell(y) = \Ay_{S^* \ell}(y^S),
\end{equation}
where $S$ is the $\ell^2$-orthogonal operator with inverse $S^* = S^{-1}$
defined by
\begin{align*}
Su(\ell)
:= \left\{ \begin{array}{ll} u(\ell), & \ell_2 > \hat{x}_2, \\
u(\ell-\burg_{12}), & \ell_2 < \hat{x}_2
\end{array} \right.
\quad {\rm and} \quad
S^* u(\ell)
:= \left\{ \begin{array}{ll} u(\ell), & \ell_2 > \hat{x}_2, \\
u(\ell+\burg_{12}), & \ell_2 < \hat{x}_2.
\end{array}\right.
\end{align*}

We can translate~\eqref{eq:perm-invariance-site-E-dislocations} to a statement
about $u_0$ and $V_\ell$. Let $S_0 w(x) = w(x), x_2 > \hat{x}_2$ and
$S_0 w(x) = w(x-\burg_{12}) - \burg, x_2 < \hat{x}_2$, then we obtain that
$y_0^S = P x + S_0 u_0$ and $S_0 u_0 \in C^\infty(\Omega_\Gamma)$ and
$S_0 (u_0 + u) = S_0 u_0 + S u$.
The permutation invariance \eqref{eq:perm-invariance-site-E-dislocations} can
now be rewritten as an invariance of $\Au_\ell(\equiv\Au_{\#},~\forall~\ell\in\L$
since $\L=A\Z^2$) under the slip $S_0$:
\begin{eqnarray}\label{slip-invariance}
\Au_{\#}\big(D(u_0+u)(\ell)\big) = \Au_{\#}\big({\sf e}(\ell)+{\sf D}u(\ell)\big)
\qquad \forall~u \in \Admu(\L),~\ell\in\Lambda
\end{eqnarray}
where
\begin{equation}\label{eq:defn_elastic_strain}
{\sf e}(\ell) := ({\sf e}_\rho(\ell))_{\rho\in\L-\ell} \quad \text{with} \quad
{\sf e}_\rho(\ell)
:= \left\{ \begin{array}{ll}
S^* D_\rho S_0u_0(\ell), & \ell\in\Omega_{\Gamma}, \\
D_\rho u_0(\ell), & \text{otherwise,}
\end{array} \right.
\end{equation}
and
\begin{equation}\label{eq:defn_Dprime}
{\sf D} u(\ell) := ({\sf D}_\rho u(\ell))_{\rho\in\L-\ell} \quad \text{with} \quad
{\sf D}_\rho u(\ell) :=
\left\{ \begin{array}{ll}
S^* D_\rho Su(\ell), & \ell \in\Omega_{\Gamma}, \\
D_\rho u(\ell), & \text{otherwise.}
\end{array} \right.
\end{equation}
The following lemma gives the decay estimate of ${\sf e}$
(see \cite{chenpre_vardef} and \cite[Lemma 3.1]{ehrlacher13}).

\begin{lemma}\label{lemma:decay_el}
	If the  predictor $u_0$ is defined by \eqref{predictor-u_0-dislocation} and
	${\sf e}(\ell)$ is given by \eqref{eq:defn_elastic_strain}, then
	there exists a constant $C$ such that
	\begin{eqnarray}\label{decay-el}
	|{\sf e}_{\sigma}(\ell)| \leq C |\sigma| \cdot |\ell|^{-1} .
	\end{eqnarray}
\end{lemma}

Similar to \eqref{def_grand_diff}, we can define grand potential difference functional
for dislocation
\begin{eqnarray}\label{def_grand_diff_disl}
\nonumber
\Gu^{\rm d}(u) &:=&
\sum_{\ell\in\L} \Big( \Gu_{\#}\big(Du_0(\ell) + Du(\ell),\muhom\big)
- \Gu_{\#}\big(Du_0(\ell),\muhom\big) \Big)
\\
&=& \sum_{\ell\in\L} \Big( \Gu_{\#}\big({\sf e}(\ell) + {\sf D}u(\ell),\muhom\big)
- \Gu_{\#}\big({\sf e}(\ell),\muhom\big) \Big),
\end{eqnarray}
where \eqref{slip-invariance} is used.
The following two lemmas are analogous to Lemma \ref{lemma:free_E_space} and
\ref{lemma:regularity} in the case of dislocations.
We refer to \cite{chenpre_vardef} (see also \cite{ehrlacher13}) for a rigorous proof.

\begin{lemma}\label{lemma:free_E_space_disl}
	If $u_0$ is given by \eqref{predictor-u_0-dislocation}, then
	$\Gu^{\rm d}$ is well-defined on $\Admu(\L)$
	and $\nu$ times Fr\'{e}chet differentiable.
\end{lemma}

\begin{lemma}\label{lemma:regularity_disl}
	If $\bar{u}\in\Admu(\L)$ is a strongly stable solution to
	\eqref{problem_min_inf} with $\Gu\equiv\Gu^{\rm d}$
	in the sense that \eqref{ubar_stability} with some constant $\bar{c}>0$,
	then there exists a constant $C>0$ such that
	\begin{equation}\label{eq:decay_Du_disl}
	|{\sf D}\bar{u}(\ell)|_{\gamma} \leq C(1+|\ell|)^{-2}\log(2+|\ell|)
	\qquad \forall \ell \in \L.
	\end{equation}
\end{lemma}

We can derive the limit of chemical potential similar to Theorem \ref{theorem:limit_mu}.

\begin{theorem} \label{theorem:limit_mu_disl}
	Let $\L=A\Z^2$, $\L_R := \L \cap B_R \uparrow \L$ and $N_R := \#\L_R$.
	For each $R$ let $u_R : \L_R \to \R^d$ with
	$y_R(\ell) := \ell + u_0(\ell) + u_R(\ell)$ a {\em configuration}
	with parameter $\mathfrak{m}$ independent of $R$.

	Let $\NeR \in \mathbb{R}$ be a prescribed number of electrons in the subsystem
	$\L_R$, chosen such that $|N_R - \NeR|$ is bounded as $R \to \infty$.
	Then, the chemical potential $\mu_R$ solving $N(y_R, \mu_R) = \NeR$
	is well-defined	and satisfies
	\begin{eqnarray}\label{eq:mu_converge_disl}
	\big|\mu_R - \muhom \big| \leq CR^{-1} .
	\end{eqnarray}
\end{theorem}

\begin{proof}
The proof is similar to that of Theorem \ref{theorem:limit_mu} in
\S~\ref{sec:proof_limit_mu}. It is only necessary to rewrite \eqref{proof:mu_limit_1}
by
\begin{eqnarray}\label{proof:NR1_B}
\nonumber
&& 
N(y_R,\mu_R)-N(y_R,\muhom)
\\\nonumber
&=& \NeR - N_{\#,R}
+ \sum_{\ell\in\L_R}\Nu_{\#}(\pmb{0},\muhom)
- \sum_{\ell\in\L_R}\Nu^{\L_R}_{\ell}\big(Du_0(\ell)+Du(\ell),\muhom\big)
\\\nonumber
&=& \big(\NeR - N_{\#,R}\big)
+ \sum_{\ell\in\L_R} \Big( \Nu_{\#}\big(\pmb{0},\muhom\big)
- \Nu_{\ell}^{\L_R}(\pmb{0},\muhom) \Big)
\\ \nonumber
&& \quad + \sum_{\ell\in\L_R}\Big( \Nu_{\ell}^{\L_R}(\pmb{0},\muhom)
- \Nu^{\L_R}_{\ell}\big({\sf e}(\ell)+{\sf D}u(\ell),\muhom\big) \Big)
\\
&=:& T_1^{\rm d}+T^{\rm d}_2+T^{\rm d}_3 ,
\end{eqnarray}
where $T_1^{\rm d}$ is uniformly bounded,
$T_2^{\rm d}$ is estimated by $|T_2^{\rm d}|\leq CR$
(analogous to \eqref{proof:mu_limit_2} with $d=2$), and $T_3^{\rm d}$ is estimated by using
Lemma \ref{lemma:locality_fixed_mu}, \ref{lemma:decay_el} and \ref{lemma:regularity_disl}
\begin{eqnarray}\label{proof:NR2_B}
|T_3^{\rm d}| \leq C\sum_{\ell\in\L_R} |{\sf e}_{\ell}+{\sf D}u(\ell)|_{\gamma} \leq CR .
\end{eqnarray}
This together with \eqref{proof:limit_mu_d_N} completes the proof.
\end{proof}

To justify the thermodynamic limits of dislocations, we define a sequence of
finite-domain equilibrium problem for dislocations:
Find $(\bar{u}_R,\bar{\mu}_R)\in\AdmR\times\R$ such that
\begin{multline}\label{equil_force_R_disl}
(\bar{u}_R,\bar{\mu}_R) \in \arg\min \Big\{ \Eu^{\L_R}(u_0+u_R, \tau) ~:~
\Nu^{\L_R}(u_0+u_R, \tau) = \NeR, ~u_R\in\AdmR \Big\} .
\end{multline}
Using the same arguments as those in \S~\ref{sec:proof_theo_1} and \S~\ref{sec:proof_theo_2},
we have the following results for dislocations, which are analogous to Theorem
\ref{theorem:limit_problem_1} and \ref{theorem:limit_problem_2}.
%

\begin{theorem}\label{theorem:limit_problem_1_disl}
	Let $\L=A\Z^2$ and $|N_R - \NeR|$ be bounded as $R \to \infty$.
	If $\bar{u}\in\Admu(\L)$ is a strongly stable solution of \eqref{problem_min_inf}
	with $\Gu \equiv \Gu^{\rm d}$
	in the sense that \eqref{ubar_stability} with some constant $\bar{c}>0$,
	then there are constants $R^{\rm d}_0, c^{\rm d}_{\rm b} > 0$
	such that, for $R>R^{\rm d}_0$ and $\Rb > c^{\rm  d}_{\rm b} \log R$,
	there exists a solution $(\bar{u}_R,\bar{\mu}_R)$ of \eqref{equil_force_R_disl}
	satisfying
	\begin{eqnarray}\label{err_Du_disl}
	\big\|{\sf D}\bar{u}-{\sf D}\bar{u}_R\big\|_{\ell^2_{\gamma}(\L_R)}
	+ \big|\bar{\mu}_R-\muhom\big| \leq CR^{-1}\log R.
	\end{eqnarray}
\end{theorem}

\begin{theorem} \label{theorem:limit_problem_2_disl}
	Let $\L=A\Z^2$, $R_j \uparrow \infty$ and
	$(\bar{u}_{R_j}, \bar{\mu}_{R_j})$ be solutions to \eqref{equil_force_R_disl},
	If $|N_{R_j} - N_{{\rm e},R_j}|$ is bounded
	and $\sup_{j > 0} \| {\sf D}\bar{u}_{R_j} \|_{\ell^2_\gamma(\L_R)} < \infty$,
	then there exists a subsequence (not relabelled) and $\bar{u} \in \Admu(\L)$
	such that
	\begin{eqnarray}\label{convergence_uj_disl}
	\bar{\mu}_{R_j} \to \muhom \quad \text{and} \quad
	{\sf D}_\rho \bar{u}_{R_j}(\ell) \to {\sf D}_\rho \bar{u}(\ell) \quad
	\forall~\ell \in \L, ~\rho \in \L - \ell
	\end{eqnarray}
	Moreover, each such accumulation point $\bar{u}$ solves \eqref{problem_min_inf}
	with $\Gu \equiv \Gu^{\rm d}$.
\end{theorem}


\begin{proof}[Proof of Theorem \ref{theorem:limit_problem_1_disl}]
	The proof is similar to that of Theorem \ref{theorem:limit_problem_1}, hence we only outline the key differences.
	We shall first redefine the notation in our dislocation setting.
	For $(u_R,\tau)\in \WR(\L)\times\R$, we define
	\begin{eqnarray*}
		\Fu^{\L_R}_{\ell}(u_R,\tau):=F_{\ell}\big((x+u_0+u_R)|_{{\L_R}}, \tau\big)
	\end{eqnarray*}
	and $\Fu(u_R,\tau)\in \WR(\L)'$ given by \eqref{eq:F_LR}.
	Similarly, let $\Nu^{\L_R}(u_R,\tau) :=  \Ny\big((x+u_0+u_R)|_{\L_R},\tau\big)$.
	Next, we define $\T_R: \WR(\L)\times\R\to \WR(\L)'\times\R$ by
	\begin{eqnarray*}
		\T_R(u_R,\tau) := \Big( -\Fu^{\L_R}(u_R,\tau) ,
		\NeR^{-1}\Nu^{\L_R}(u_R,\tau)-1 \Big)
		\quad{\rm for}~(u_R,\tau)\in \WR(\L)\times\R.
	\end{eqnarray*}
	Then \eqref{equil_force_R_disl} is again equivalent to
	\begin{eqnarray}\label{equil_T_D}
	\T_R(\bar{u}_R,\bar{\mu}_R)=\pmb{0}.
	\end{eqnarray}
	The Jacobian matrix $\J_R(u_R,\tau)$ of $\T$ at  $(u_R,\tau)$ is still defined  by \eqref{Jac_T}.

	{\it Step 1. Quasi-best approximation.}
	We construct $T_R\bar{u}\in\AdmR$ such that for $R$ sufficiently large,
	\begin{eqnarray*}\label{proof-B-3-1}
		\|{\sf D}T_R\bar{u} - {\sf D}\bar{u}\|_{\ell^2_\gamma}
		\leq C \|{\sf D}\bar{u}\|_{\ell^2_\gamma(\Lambda\backslash B_{R/2})}
		\leq CR^{-1}\ln R ,
	\end{eqnarray*}
	where Lemma \ref{lemma:regularity_disl} is used for the last inequality.
	It is easy to see that $\delta\Gu$ and $\delta^2\Gu$ are locally Lipschitz
	continuous, hence
	\begin{eqnarray}\label{proof:appr_L_B}
	\|\delta\Gu(\bar{u})-\delta\Gu(T_R\bar{u})\| \leq CR^{-1}\ln R
	\quad {\rm and} \quad
	\|\delta^2\Gu(\bar{u})-\delta^2\Gu(T_R\bar{u})\| \leq  CR^{-1}\ln R.
	\quad
	\end{eqnarray}
	in a neighbourhood or $\bar{u}$.

	{\it Step 2. Consistency.}
	Let $\Gu^{\L_R}(u,\tau) := \Gy\big((x+u_0+u)|_{\L_R},\tau\big)$.
	Same as \eqref{proof:consistency_b}, we have
	\begin{equation}\label{proof:consistency_b_D}
	\left| \big\< \delta_u\Gu^{\L_R}(T_R\bar{u},\muhom)
	- \delta_u\Gu(T_R\bar{u},\muhom) ,v \big> \right|
	\leq Ce^{-\gamma_{\rm c}\Rb} R^{3/2}\|Dv\|_{\ell^2_\gamma},
	\end{equation}
	and can choose $\Rb$ sufficiently large such that $e^{-\gamma_{\rm c} \Rb} R^{3/2} \leq C R^{-1}\ln R$.
	Then we obtain from \eqref{proof:appr_L_B}, \eqref{proof:consistency_b_D} and $\delta\Gu(\bar{u}) = 0$
	that $\forall~v\in \WR(\L)$,
	\begin{eqnarray}\label{proof:consistency_c_D}
	\nonumber
	&& \big\< -\Fu^{\L_R}(T_R\bar{u},\muhom),v \big>
	~=~ \big\<- \delta_u\Gu^{\L_R}(T_R\bar{u},\muhom),v\big\>
	\\ \nonumber
	&=& \big\<\delta_u\Gu(T_R\bar{u},\muhom)
	- \delta_u\Gu^{\L_R}(T_R\bar{u},\muhom),v\big\> +
	\big\< \delta_{u}\Gu(\bar{u},\muhom)-\delta_u\Gu(T_R\bar{u},\muhom) , v \big\>
	\\
	&\leq& C\big(e^{-\gamma_{\rm c}R} R^{3/2}
	+ R^{-1}\ln R  \big) \cdot \|Dv\|_{\ell^2_\gamma}
	~\leq~ CR^{-1}\ln R \cdot \|Dv\|_{\ell^2_\gamma}
	\end{eqnarray}
	for sufficiently large $R$ and $R_{\rm b}$.

	To estimate the residual of
	$\NeR^{-1}\Nu^{\L_R}(T_R\bar{u},\muhom)-1$,
	we construct a corresponding homogeneous finite system
	$\L_R = \Lhom\cap B_{R+R_{\rm b}}$ with $N_{\#,\L_R}$ electrons,
	and then obtain from an argument analogous to
	\eqref{proof:NR1_B}-\eqref{proof:NR2_B} that
	\begin{align*}
		\big|\Nu^{\L_R}(T_R\bar{u},\muhom)-\NeR\big|
		&\leq \Big| \Nu^{\L_R}(T_R\bar{u},\muhom)
		- \sum_{\ell\in\L_R}\Nu_{\#}(\pmb{0},\muhom) \Big|
		+ \big|N_{{\rm e},R} - N_{\#,\L_R}\big| \\
		&\leq CR,
	\end{align*}
	where $C$  is independent of $R$.
	Therefore, we have
	\begin{eqnarray}\label{proof:consistency_d_D}
	\big|\NeR^{-1}\Nu^{\L_R}(T_R\bar{u},\muhom)-1\big|
	\leq CR^{-1},
	\end{eqnarray}
	hence from
	\eqref{proof:consistency_c_D} and \eqref{proof:consistency_d_D} we obtain consistency estimate
	\begin{eqnarray}\label{proof:consistency_D}
	\left\|\T_R\big(T_R\bar{u},\muhom\big)\right\|_{\WR(\L)'\times\R}
	\leq CR^{-1}\ln R .
	\end{eqnarray}

	{\it Step 3. Stability.}
	Using the same proof as that for the stability result in Theorem \ref{theorem:limit_problem_1}, we can show that
	\begin{eqnarray}\label{proof:stability_D}
	\J_R\big(T_R\bar{u},\muhom\big): \WR(\L)\times\R\rightarrow \WR(\L)'\times\R ~\text{ is an isomorphism.}
	\end{eqnarray}

	{\it Step 4. Application of Inverse Function Theorem.}
	With the consistency \eqref{proof:consistency_D} and the stability \eqref{proof:stability_D},
	we can apply the inverse function theorem \cite[Lemma B.1]{luskin13}
	on fuction $\T_R$ at the point $\big(T_R\bar{u}, \muhom\big)$,
	to obtain the existence of $\bar{u}_R$ and the estimate \eqref{err_Du_disl}.
\end{proof}

\begin{proof}[Proof of Theorem \ref{theorem:limit_problem_2_disl}]
	Following the proof of Theorem \ref{theorem:limit_problem_2} verbatim (only by replacing $D\equiv{\sf D}$)  gives the desired result.
\end{proof}

\bibliographystyle{siam}
\bibliography{tdlimitbib}

\end{document}